\documentclass[11pt]{article}

\def\bSig\mathbf{\Sigma}

\newcommand{\indep}{\perp\!\!\!\perp}
\newcommand{\simresult}[3]{#1 & #2 & #3}
\newcommand{\simtableheader}{%
\toprule
& \multicolumn{3}{c}{$\delta(1)$}
& \multicolumn{3}{c}{$\delta(0)$}
& \multicolumn{3}{c}{$\zeta(1)$}
& \multicolumn{3}{c}{$\zeta(0)$} \\
\cmidrule(lr){2-4}\cmidrule(lr){5-7}\cmidrule(lr){8-10}\cmidrule(lr){11-13}
Setting &
$N$ & $N_{\mathrm{sim}}$ & Power &
$N$ & $N_{\mathrm{sim}}$ & Power &
$N$ & $N_{\mathrm{sim}}$ & Power &
$N$ & $N_{\mathrm{sim}}$ & Power \\
\midrule
}
\usepackage[a4paper,margin=1in]{geometry}
\usepackage{amsmath,amssymb,amsfonts,amsthm}
\usepackage{bm}
\usepackage{url}
\usepackage{graphicx}
\usepackage{booktabs}
\usepackage{array}
\usepackage{placeins}
\usepackage{float}
\usepackage{natbib}
\usepackage[colorlinks=true,linkcolor=blue,citecolor=blue,urlcolor=blue]{hyperref}

\newtheorem{assumption}{Assumption}
\newtheorem{proposition}{Proposition}
\newtheorem{theorem}{Theorem}
\newtheorem{lemma}{Lemma}
\newtheorem{corollary}{Corollary}
\newenvironment{keywords}{\par\medskip\noindent\textbf{Keywords:} }{\par\medskip}

\title{Power and sample size calculations for causal mediation analysis with a binary mediator in randomized trials}
\author{Bosen Cui$^{1}$, 
	Yuhong Yang$^{1,2}$, and Fan Yang$^{1,2,*}$ \\
	$^{1}$Yau Mathematical Sciences Center, Tsinghua University, Beijing, China \\
	$^{2}$Yanqi Lake Beijing Institute of Mathematical Sciences and Applications, Beijing, China \\
	\texttt{yangfan1987@tsinghua.edu.cn}}
\date{}
\begin{document}

\label{firstpage}

\maketitle

\begin{abstract}
Mediation analyses are increasingly conducted in randomized trials, but a sample size adequate for the total treatment effect may leave the natural indirect effect (NIE) or natural direct effect (NDE) substantially underpowered. Randomization does not extend to the mediator, so precision depends on the conditional mediator distribution and the mediator-outcome association, neither of which enters a total-effect calculation. Planning outside linear structural equation models is largely based on simulation under a fully specified data-generating mechanism rarely available at the design stage.  This paper develops analytic power and sample size formulas for the NIE and NDE with a binary mediator and a continuous or binary outcome. Under standard identification assumptions, we focus on the ratio-of-mediator-probability weighting (RMPW) estimator that does not require an outcome model for effect estimation. We decompose the oracle variances of the RMPW estimators into components capturing mediator-probability-ratio variability, outcome variation, and their association, with an additional shared-arm covariance term for the NIE. Under a probit latent-index mediator and a working outcome model, these components are determined by a small number of interpretable design inputs rather than by the full joint distribution of covariates, mediator, and outcome.
Simulations show that the analytic sample sizes closely match simulation-based benchmarks, attain the target power, and maintain type I error near the nominal level. An ACTG175 illustration shows how pilot data can calibrate the inputs.
\end{abstract}

\begin{keywords}
Causal mediation analysis; Natural direct effect; Natural indirect effect; Ratio-of-mediator-probability weighting; Sample size calculation; Study design.
\end{keywords}
\section{Introduction}
Understanding how an intervention produces its effect is often an important scientific objective in a randomized trial. Causal mediation analysis provides a formal framework for this objective by distinguishing effects transmitted through a mediator from those operating through other pathways \citep{lee2021guideline}. In many clinical, behavioral, and policy applications, the mediator of interest is binary, indicating, for example, whether a participant adheres to treatment, achieves an early response, or experiences an intermediate clinical event \citep{ralston2014home,hong2015ratio,stensrud2017diastolic, vandenberghe2018surrogate,heerspink2022pre,rijnhart2023statistical}. When these pathways are prespecified scientific targets, adequate power for the corresponding natural indirect effect (NIE) or natural direct effect (NDE) must be considered when the trial is designed.

Power and sample size theory for the total treatment effect, i.e., the average treatment effect (ATE), is well developed, covering classical randomized-trial calculations, covariate-adjusted estimators and observational settings \citep{lachin1981introduction,dupont1990power,wittes2002sample,chow2017sample,schuler2022designing,shook2022power,branson2024power,liu2026sample}. These results, however, do not determine the sample size required for a mediation effect: on an additive scale, the same ATE can correspond to markedly different decompositions into the NIE and NDE. Even for a fixed target NIE or NDE, its precision depends on features not contained in conventional ATE calculations, including the treatment-specific conditional mediator distributions and the mediator--outcome relationship given baseline covariates in each treatment arm. The underlying difficulty is that treatment randomization does not extend to the mediator. Identification of natural effects therefore requires assumptions beyond randomization \citep{robins2003semantics,imai2010general,imai2010identification}, and their power calculations require additional design information about the mediator and mediator--outcome processes. 

Existing power and sample size methods for mediation analysis fall broadly into two classes. Earlier work studied the \citet{baron1986moderator} estimator and related tests under linear structural equation models, providing model-specific calculations or tabulated sample-size recommendations \citep{ mackinnon2002comparison,fritz2007required,vittinghoff2009sample}. These methods, however, formulate the indirect effect as a product of coefficients, which has an NIE interpretation only under additional linearity assumptions. In particular, with treatment--mediator interaction in the outcome model or nonlinear links for binary mediators or outcomes, the coefficient product no longer identifies the NIE. \cite{kelcey2017statistical} provide an important extension, allowing treatment--mediator interaction in closed-form power calculations for indirect effects in group-randomized trials. {Their analysis, however, is confined to indirect effects and retains linear mixed-effects models for continuous mediators and outcomes, and it does not cover the binary-mediator setting or the binary-outcome extension considered here.} For more general mediation settings, broadly applicable power calculations rely on Monte Carlo simulation under prespecified mediator and outcome models \citep{thoemmes2010power,schoemann2017determining,qin2024sample}. Simulation is flexible but requires specifying a detailed joint data-generating mechanism for the baseline covariates, mediator, and outcome. Such specifications are often weakly informed at the design stage. Moreover, simulation alone offers little guidance on which design-stage summaries suffice: distinct data-generating mechanisms may agree on selected marginal summaries yet imply substantially different estimator variances and sample-size requirements.

 In summary, existing approaches to power and sample size calculation either impose linear specifications of the mediator and outcome models that do not accommodate a binary mediator, or require a detailed joint data-generating mechanism that is rarely well informed at the design stage. In this paper, we address this limitation by starting from the sampling variability of a prespecified causal mediation estimator. Specifically, we consider ratio-of-mediator-probability weighting (RMPW) estimators of the NIE and NDE with a binary mediator in randomized trials \citep{hong2015ratio,bein2018two}.  { RMPW extends to mediation the role that inverse-probability-of-treatment weighting plays for the ATE and estimates the mediation effects without an outcome model.} Building on the oracle large-sample variance of the RMPW estimators, we decompose the arm-specific variance of each weighted outcome mean into three components: mediator-probability-ratio variability, outcome variation, and their association. For NIE, the calculation additionally requires the covariance between the two weighted means estimated from the same treatment arm. Under a probit latent-index mediator and a working outcome model, we then express these variance components in terms of interpretable design inputs. These inputs describe  mediator prevalences and covariate-driven mediator heterogeneity, outcome variance or event probabilities, pathway strength, and the association between the outcome and a baseline-covariate summary. Together with the target effect, treatment allocation, and error rates, they yield Wald-type power and sample size calculations for continuous and binary outcomes.

The remainder of the paper is organized as follows. Section~\ref{sec:setup} {reviews} the
causal mediation estimands, identification assumptions, and weighting
representation. Section~\ref{sec:variance-decomposition} derives a variance decomposition for the RMPW estimator with a binary mediator, and connects the
variance to power and sample size. Section~\ref{sec:binary-working-model} introduces the working model used
to evaluate the variance components at the design stage and links the power function with interpretable inputs. Section~\ref{sec:numerical} examines
the accuracy and sensitivity of the proposed calculations through simulation
studies and a design illustration based on ACTG175 pilot data. {Section~\ref{sec:conclusion} concludes. The Supplementary Material contains proofs and additional derivations; it also provides sample size calculations for both the NIE and NDE under linear mediator and outcome models with a continuous mediator, allowing treatment--mediator interaction in single-level data, thereby complementing the linear-model results discussed above.}

\section{Causal mediation estimands and the RMPW estimator}
\label{sec:setup}

Consider a randomized trial with an i.i.d.\ sample of size $N$. For each participant, let $O=(\mathbf X,A,M,Y)$, where $\mathbf X\in\mathcal X\subseteq\mathbb R^p$ is a vector of baseline covariates, $A\in\{0,1\}$ is the randomized treatment assignment with known allocation probabilities $\pi_a=P(A=a)\in (0,1)$, $M\in\{0,1\}$ is a binary mediator, and $Y\in\mathbb R$ is a continuous or binary outcome. Let $M(a)$ denote the potential mediator value under treatment level $a$,
and let $Y(a,m)$ denote the potential outcome if
treatment were set to $a$ and the mediator to $m$.
The observed mediator and outcome satisfy the consistency assumptions
$M=M(A)$ and $Y=Y(A,M(A))$. For $a,a'\in\{0,1\}$, define the mediation
functional
$
\theta(a,a')=E\{Y(a,M(a'))\}.
$
The natural indirect effect (NIE) under treatment level $a$ is
\[
\delta(a)=\theta(a,1)-\theta(a,0),
\]
and the natural direct effect (NDE) at mediator distribution level $a$ is
\[
\zeta(a)=\theta(1,a)-\theta(0,a).
\]
Thus, $\delta(a)$ varies the mediator from $M(0)$ to $M(1)$ while fixing
the treatment at $a$, whereas $\zeta(a)$ varies the
treatment from 0 to 1 while holding the mediator at $M(a)$.

Because the mediator is not randomized, identification of $\theta(a,a')$ requires assumptions beyond randomization of treatment assignment. We invoke sequential ignorability and positivity assumptions \citep{imai2010general,imai2010identification}.

	\begin{assumption}[Sequential ignorability and positivity] \label{seq ign}
For $a,a'\in\{0,1\}$ and $m\in\{0,1\}$:
(i) $\{Y(a',m),M(a)\}\indep A$; 
(ii) $
Y(a,m)\indep M(a')\mid A=a',\mathbf X 
$; and
(iii) 
$0<p_a(\mathbf X)<1$ almost surely, where
$
p_a(\mathbf x)=P(M=1\mid A=a,\mathbf X=\mathbf x).
$
	\end{assumption}
Condition (i) holds by randomization; (ii) requires no unmeasured pretreatment and no post-treatment mediator--outcome confounding given $\mathbf X$; (iii) requires both mediator levels to occur with positive probability at every covariate value within each treatment arm.

Among the equivalent identification strategies for $\theta(a,a')$, we develop
the sample size calculations based on the ratio-of-mediator-probability weighting (RMPW) estimator \citep{hong2015ratio,bein2018two}, a commonly used weighting estimator for mediation effects that plays a role analogous to inverse-probability-of-treatment weighting for the ATE. { An equivalent weighting representation, with weights expressed through treatment propensities rather than mediator probabilities, is developed by \citet{huber2014identifying}}. Alternative estimators based on regressions or semiparametric augmentation \citep{tchetgen2012semiparametric} would require, for their design-stage variance, prior information about outcome-model parameter values that is difficult to obtain before data collection, whereas, as shown in Section \ref{sec:binary-working-model}, the RMPW calculation requires weaker outcome-model inputs. Writing $
g_a(m\mid \mathbf x)
=
p_a(\mathbf x)^m\left( 1-p_a(\mathbf x)\right)^{1-m}
$ and 
$
R_{a,a'}(m,\mathbf x)
=
\frac{g_{a'}(m\mid \mathbf x)}
{g_a(m\mid \mathbf x)}
$, under
Assumption~\ref{seq ign},
\[
\theta(a,a')
=
E\left[
\frac{I(A=a)}{\pi_a}
R_{a,a'}(M,\mathbf X)Y
\right],
\]
where the ratio $R_{a,a'}$ reweights participants in arm $a$ so that, given covariates, their mediator distribution represents that under arm $a'$; when $a'=a$, $R_{a,a'}(M,\mathbf X)=1$, and the estimand reduces to the arm-$a$ outcome mean. The normalized RMPW estimator is  
\begin{equation*}
	\hat\theta(a,a')
	=
	\frac{
		\sum_{i=1}^N \mathbb{I}(A_i=a)R_{a,a'}(M_i,\mathbf X_i)Y_i
	}{
		\sum_{i=1}^N \mathbb{I}(A_i=a)R_{a,a'}(M_i,\mathbf X_i)
	},
\end{equation*}
with $\hat\delta(a)=\hat\theta(a,1)-\hat\theta(a,0)$ and
$\hat\zeta(a)=\hat\theta(1,a)-\hat\theta(0,a)$.
For the design stage calculations, we treat the mediator probabilities as
known and use the resulting oracle estimator as a tractable benchmark. In
practice, these probabilities are estimated from the data, and because the oracle and estimated-ratio variances have
no general ordering \citep{bein2018two}, we examine the estimated-ratio implementation in simulations. In the settings considered, the estimated-ratio implementation yields empirical power similar to the oracle-ratio benchmark.

\section{Variance decomposition and sample size calculation}
\label{sec:variance-decomposition}

This section derives the variance expressions underlying the design calculations, starting from the oracle large-sample variance of the RMPW estimator established by \cite{bein2018two}, in which the mediator probability ratios are fixed at their population values. 

\begin{proposition}[Oracle variance of the RMPW estimator]
\label{prop:oracle-rmpw-var}
Suppose Assumption~\ref{seq ign} holds and  $E\{R_{a,a'}^2(M,\mathbf X)(1+Y)^2\mid A=a\}<\infty$. For the oracle normalized RMPW estimator,
\[
\sqrt N\{\hat\theta(a,a')-\theta(a,a')\}
=
\frac{1}{\sqrt N}\sum_{i=1}^N\psi_{a,a'}(O_i)+o_p(1),
\]
where
$
\psi_{a,a'}(O)
=
\frac{\mathbb{I}(A=a)}{\pi_a}R_{a,a'}(M,\mathbf X)
\{Y-\theta(a,a')\}.
$
Consequently, 
\begin{equation*}
	\sqrt{N} \{\hat\theta(a,a')-\theta(a,a')\} \stackrel{d}{\to} N(0, V_\theta(a,a')),
\end{equation*}
where
\begin{equation}
V_\theta(a,a')
=
\frac{1}{\pi_a}
E\left[
R_{a,a'}^2(M,\mathbf X)
\{Y-\theta(a,a')\}^2
\mid A=a
\right].
\label{eq:rmpw-theta-var}
\end{equation}
\end{proposition}

\textit{The natural direct effect.} To understand the variance of the NDE estimator, first consider the building block $V_\theta(a,a')$ in Equation~\eqref{eq:rmpw-theta-var}. When $a'=a$,
$R_{a,a}(M,\mathbf X)=1$, and the formula reduces to the variance of
an arm-specific sample mean. When $a'\neq a$, the variance additionally depends on the distribution of the mediator weights and their relationship with the outcome. To separate these contributions, define
$
d_{a,a'}=E[R_{a,a'}^2(M,\mathbf X)\mid A=a],
$
$
B_{a,a'}=E[\{Y-\theta(a,a')\}^2\mid A=a],
$
and
$
C_{a,a'}=
\operatorname{Cov}\left[
R_{a,a'}^2(M,\mathbf X),
\{Y-\theta(a,a')\}^2
\mid A=a
\right].
$
Then \eqref{eq:rmpw-theta-var} can be written as
\begin{equation}
V_\theta(a,a')
=
\frac{1}{\pi_a}
\left\{
 d_{a,a'}B_{a,a'}+C_{a,a'}
\right\}.
\label{eq:theta-design-effect-decomp}
\end{equation}
For a binary mediator, 
\[
d_{a,a'}
=
E\left[
\frac{p_{a'}^2(\mathbf X)}{p_a(\mathbf X)}
+
\frac{\{1-p_{a'}(\mathbf X)\}^2}{1-p_a(\mathbf X)}
\right] = 1+ E \left[\frac{\left[p_{a'}(\mathbf{X}) - p_a(\mathbf{X})\right]^2}{p_a(\mathbf{X}) (1-p_a(\mathbf{X})) }\right]\geq 1,
\]
with equality if and only if $p_{a'}(\mathbf{X})=p_{a}(\mathbf{X})$ almost surely, i.e., the two conditional mediator distributions coincide; $d_{a,a'}$ is large when the two arms disagree on the conditional mediator probabilities in covariate regions where $p_a(\mathbf{X})$ is near 0 or 1. The term $B_{a,a'}$ is an outcome second moment around $\theta(a,a')$, and $C_{a,a'}$ captures how mediator reweighting concentrates observations with different outcome variability. Directly specifying $C_{a,a'}$ at the design stage is difficult; its Cauchy-Schwarz bound gives a conservative alternative but discards both the sign of $C_{a,a'}$ and the dependence structure that determines its magnitude.

The NDE $\zeta(a)=\theta(1,a)-\theta(0,a)$ contrasts two mediation functionals whose RMPW estimators use observations from different treatment arms. The two influence functions contain mutually exclusive treatment indicators and therefore have zero covariance. Consequently,
\begin{equation}
V_\zeta(a)
=
\sum_{j=0}^1
\frac{1}{\pi_j}
\left\{
 d_{j,a}B_{j,a}+C_{j,a}
\right\}= 
\frac{B_{a,a}}{\pi_a}
+
\frac{1}{\pi_{1-a}}
\left\{
d_{1-a,a}B_{1-a,a}+C_{1-a,a}
\right\}.
\label{eq:direct-design-effect-var}
\end{equation}
The first term on the right-hand side involves no mediator reweighting, whereas
the second term accounts for transporting the mediator distribution from arm
$a$ to arm $1-a$. 

\textit{The natural indirect effect.} Unlike the NDE, the NIE $\delta(a)=\theta(a,1)-\theta(a,0)$ cannot be handled by simply adding two component variances: its RMPW estimators $\hat{\theta}(a,1)$ and $\hat{\theta}(a,0)$ use the same treatment-arm observations, so their influence functions overlap and their covariance enters the variance. Therefore, the asymptotic variance of $\hat{\delta}(a)$ is
\begin{equation}
\begin{split}
V_\delta(a)
=
\frac{1}{\pi_a}
E\Big[&
\{
R_{a,1}(M,\mathbf X)\{Y-\theta(a,1)\}
-
R_{a,0}(M,\mathbf X)\{Y-\theta(a,0)\}
\}^2
\mid A=a
\Big].
\end{split}
\label{eq:indirect-general-var}
\end{equation}
Equivalently,
\begin{equation}
V_\delta(a)
=
\frac{1}{\pi_a}
\left\{
 d_{a,1}B_{a,1}+C_{a,1}
+d_{a,0}B_{a,0}+C_{a,0}
-2K_a
\right\},
\label{eq:indirect-design-effect-var}
\end{equation}
where
$ K_a
=
\operatorname{Cov}\left[
R_{a,1}(M,\mathbf X)\{Y-\theta(a,1)\},
R_{a,0}(M,\mathbf X)\{Y-\theta(a,0)\}
\mid A=a
\right]
$ is the covariance between the weighted residual terms
associated with $\theta(a,1)$ and $\theta(a,0)$, which does not
vanish in general because both terms are based on the same observations from treatment
arm $a$. As shown in the Supplementary Material, $K_a\geq 0$, so omitting the term $-2K_a$ yields a
conservative upper bound for $V_\delta(a)$; combining this omission with separate Cauchy--Schwarz bounds for
$C_{a,1}$ and $C_{a,0}$, however, is even more conservative and can substantially overstate the required sample size. Because $C_{a,a'}$ and
$K_a$ involve joint weight--outcome behavior that is difficult to elicit directly, Section~\ref{sec:binary-working-model} evaluates them
under working models in terms of interpretable design
inputs.

\textit{Wald approximation.} Let $\tau(a)$ denote $\zeta(a)$ or $\delta(a)$,
with design-stage variance $V_\tau(a)$ on the $\sqrt N$ scale. Consider a one-sided level $\alpha$ test of $H_0: \tau (a) =0$ against $H_1: \tau(a) >0$. We reject $H_0$ when $\hat\tau (a)>
z_{1-\alpha}
\sqrt{\frac{V_\tau(a)}{N}}.$
 Under a prespecified alternative $\tau(a)=\tau_*(a)>0$ with variance $V_{\tau_*}(a)$, the standard Wald approximation gives 
\[
\operatorname{Power}(N;\tau_*(a))
\approx
\Phi\left(
\frac{\sqrt N\tau_*(a)}{\sqrt{V_{\tau_*}(a)}}
-z_{1-\alpha}
\right),
\]
so the required sample size $N$ to reach power $1-\beta$ for a level-$\alpha$ test is
\begin{equation}
N
\approx
\left\lceil
\frac{(z_{1-\alpha}+z_{1-\beta})^2V_{\tau_*}(a)}{\tau_*^2(a)}
\right\rceil .
\label{eq:wald-sample-size}
\end{equation}
 The
main task is therefore to specify $V_{\tau_*}$ at the design stage, which the next section addresses.
Proofs of Proposition~\ref{prop:oracle-rmpw-var},
Equations~\eqref{eq:direct-design-effect-var}--\eqref{eq:wald-sample-size},
and the result $K_a\geq 0$ are given in Section~\ref{app:proof-oracle-rmpw-placeholder} of the Supplementary Material.

\section{Design-stage calculation}
\label{sec:binary-working-model}
This section expresses the variance components $d_{a,a'}$, $B_{a,a'}$, $C_{a,a'}$ and $K_a$ of Section~\ref{sec:variance-decomposition}  in terms of prespecified
planning quantities under working models for the mediator and the outcome. We first evaluate $d_{a,a'}$
under a binary-mediator working model, which does not depend on the outcome type, and then evaluate the remaining components separately for continuous and binary outcomes. Detailed derivations of the results in this section are provided in Section~\ref{app:main-text-proofs} of the Supplementary Material.

Consider a probit latent-index representation for the binary mediator,
\begin{equation*}
M(a)=\mathbb{I}\{f_M(\mathbf X)+a\Delta_M+\epsilon_M>0\},
\end{equation*}
where $\epsilon_M\sim \mathcal{N}(0,1)$ is independent of $\mathbf X$, 
$f_M(\mathbf X)$ is the baseline component of the latent mediator propensity,
and $\Delta_M$ is the treatment-induced shift. Let $W=f_M(\mathbf X)$ denote this baseline-covariate summary. When $f_M(\mathbf{X})$ is a linear combination of standardized baseline covariates satisfying the Lyapunov-type condition in Lemma~\ref{lem:lyapunov-score} of the Supplementary Material, which in essence requires that several covariates contribute to the index and none dominates it, we approximate the marginal distribution of $W$ by $N(\mu_W,\sigma^2_W)$. The approximation is plausible if the latent index aggregates many modest contributions rather than being driven by a single dominant covariate, or if any dominant covariates are approximately normally distributed. The former is common in practice: to make the sequential ignorability assumption credible, 
$\mathbf{X}$ is typically a rich collection of baseline covariates. This approximation is a design-stage device for translating planning inputs into variance components; the variance results of Section~\ref{sec:variance-decomposition} do not depend on it. Section~\ref{sec:numerical} confirms that the resulting calculations remain accurate under markedly non-Gaussian covariates.

Under the latent-index representation, the conditional mediator probability is $p_a^W(w)
=
P\{M(a)=1\mid W=w\}
=
\Phi(w+a\Delta_M),
a=0,1,$ where $\Phi(\cdot)$ denotes the standard normal distribution function. Rather than specifying these latent
parameters $(\mu_W, \sigma_W^2, \Delta_M)$ directly, we parameterize the working model through three design inputs:
the marginal mediator probabilities $q_a=P\{M(a)=1\}, a=0,1$, and
$
R_M^2=\frac{\sigma_W^2}{1+\sigma_W^2},
$
the proportion of latent mediator variation explained by baseline covariates. These inputs determine the working
parameters:
\begin{equation}
\sigma_W^2=\frac{R_M^2}{1-R_M^2},
\qquad
\mu_W=\frac{\Phi^{-1}(q_0)}{\sqrt{1-R_M^2}},
\qquad
\Delta_M=\frac{\Phi^{-1}(q_1)-\Phi^{-1}(q_0)}
{\sqrt{1-R_M^2}} .
\label{eq:binary-score-parameters}
\end{equation}

\textit{Evaluation of $d_{a,a'}$.} Define the working Bernoulli mass function and mediator probability ratio by
$g_a^W(m\mid w)
	=
	\{p_a^W(w)\}^m
	\{1-p_a^W(w)\}^{1-m}$ and $R_{a,a'}^W(m,w)
	=
	\frac{
		g_{a'}^W(m\mid w)
	}{
		g_a^W(m\mid w)
	}.
$
Under the working model,
\begin{equation*}
	\begin{aligned}
		d_{a,a'}
		&=
		E_W\left[
		\sum_{m=0}^1
		g_a^W(m\mid W)
		\{R_{a,a'}^W(m,W)\}^2
		\right]=
		1+
		E_W\left[
		\frac{
			\{p_{a'}^W(W)-p_a^W(W)\}^2
		}{
			p_a^W(W)\{1-p_a^W(W)\}
		}
		\right],
	\end{aligned}
\end{equation*}
where $E_W$ denotes expectation over $W\sim\mathcal{N}(\mu_W,\sigma_W^2)$.
Given
$(q_0,q_1,R_M^2)$, this is a one-dimensional Gaussian integral that can be evaluated numerically.

%The remainder of this section treats continuous and binary outcomes
%separately. For each outcome type, we first specify a working outcome model
%and calibrate its parameters from outcome-specific planning quantities. We
%then evaluate the remaining variance components $B_{a,a'}$, $C_{a,a'}$, and $K_a$, substitute them into the
%NDE and NIE variance decompositions of
%Section~\ref{sec:variance-decomposition}, and identify the inputs required
%for the corresponding power and sample size calculations.
\subsection{Continuous outcome}
For a continuous outcome, the key step is to project the prognostic component of the potential outcome onto $W$. Consider the working outcome model
\begin{equation*}
Y(a,m)=f_a(\mathbf X)+\gamma_a m+\epsilon_a,
\qquad
E(\epsilon_a\mid \mathbf X, M(a))=0,
\end{equation*}
where $f_a(\mathbf X)$ is the prognostic component and $\gamma_a$ describes the mediator--outcome pathway strength; both may differ across treatment arms. We further assume homoscedasticity $\operatorname{Var}\{\epsilon_a\mid W,M(a)\} = \sigma_{\epsilon,a}^2 < \infty,$
and adopt a joint Gaussian working approximation for
$\{f_a(\mathbf X),W\}$. The approximation is exact for jointly Gaussian covariates with linear $f_a(\mathbf X)$ and $W$ and is otherwise motivated by a multivariate Lyapunov argument (Section~\ref{app:lyapunov-score} of the
Supplementary Material).

The association between $f_a(\mathbf X)$ and
$W$ is summarized by the population linear projection. For $R_M^2>0$, define $b_a= \frac{
	\operatorname{Cov}\{f_a(\mathbf X),W\}
}{
	\sigma_W^2
}$ and write
$
f_a(\mathbf X)
=
E\{f_a(\mathbf X)\}+b_a(W-\mu_W)+r_a
$. Under the joint Gaussian approximation, the residual $r_a$ satisfies $E(r_a)=0$ and is independent of $W$, and hence of $\{M(a),W\}$, because $M(a)$ depends on $\mathbf{X}$ only through $W$ and an independent latent error. When $R_M^2=0$, $W$ is degenerate, $b_a$ is not defined by the ratio above, and the
projection term $b_a(W-\mu_W)$ is omitted.

The outcome model can then be written as 
\begin{equation}
Y(a,m)= E\{f_a(\mathbf X)\}+ \gamma_a q_a + b_a(W-\mu_W)+\gamma_a(m-q_a)+v_a,
\label{eq:continuous-outcome-projection}
\end{equation}
where $v_a = r_a +\epsilon_a$ satisfies 
$
E\{v_a\mid W,M(a)\}=0$ and
$
E\{v_a^2\mid W,M(a)\} = \operatorname{Var}(r_a)+\sigma_{\epsilon,a}^2:=\sigma_{v,a}^2.$ These moments follow from the working 
models rather than being additional restrictions, and no distribution is imposed on $v_a$ beyond them.

Under \eqref{eq:continuous-outcome-projection}, the NIE and NDE are
\[
\delta(a)=\gamma_a(q_1-q_0),
\qquad
\zeta(a)=E\{f_1(\mathbf X)\}-E\{f_0(\mathbf X)\}+(\gamma_1-\gamma_0)q_a.
\]
The NIE is a pure pathway product: $q_1-q_0$ is the treatment-induced
change in mediator prevalence and $\gamma_a$ the mediator--outcome
pathway strength. The target NIE $\delta(a)$ therefore calibrates $\gamma_a$ and thereby enters the variance components. The NDE additionally
involves $E\{f_1(\mathbf X)\}-E\{f_0(\mathbf X)\}$, which is unconstrained by the mediator pathway and enters none of the variance component; the target NDE $\zeta(a)$ is therefore specified
as input, and affects
the calculation only through the effect size in \eqref{eq:wald-sample-size}. When an ATE target is also prespecified, as is typical in trial planning, the NDE target is not free but implied by the decomposition $\zeta(a)=\mathrm{ATE}-\delta(1-a)$, and should be specified accordingly.

\textit{Evaluation of $B_{a,a'}$.} Recall $B_{a,a'} =E[\{Y-\theta(a,a')\}^2\mid A=a]$, which measures outcome variation in arm $a$ around $\theta(a,a')$. Because $Y = Y(a,M(a))$ given $A=a$, this second moment decomposes into the variance of $Y(a,M(a))$ and the squared deviation of its mean $\theta(a,a)$ from $\theta(a,a')$. We therefore take as a design input the marginal outcome variance $S_a^2=\operatorname{Var}\{Y(a,M(a))\}
$, the standard variance input of conventional sample size calculations. Under the projected model \eqref{eq:continuous-outcome-projection} , 
\begin{equation*}
	B_{a,a'}
	=
	S_a^2
	+
	\gamma_a^2(q_a-q_{a'})^2.
\end{equation*}
Thus, $B_{a,a}=S_a^2$, and for $a'\ne a$ the additional term reflects the mediator prevalence difference between the observed arm and the target.

\textit{Evaluation of  $C_{a,a'}$.} Recall $C_{a,a'}
=
\operatorname{Cov}\left[
R_{a,a'}^2(M,\mathbf X),
\{Y-\theta(a,a')\}^2
\mid A=a
\right].$
 Define the centered conditional mean 
	$h_{a,a'}(m,w)
	= E\{Y-\theta(a,a')\mid W=w,M=m,A=a\} = b_a(w-\mu_W)
	+
	\gamma_a(m-q_{a'}),
$
where the second equality holds under model \eqref{eq:continuous-outcome-projection}. The mean-zero and homoscedasticity properties of $v_a$ then give
\begin{equation}
	C_{a,a'}
	=
	E_W\left[
	\sum_{m=0}^1
	g_a^W(m\mid W)
	\left\{
	\{R_{a,a'}^W(m,W)\}^2-d_{a,a'}
	\right\}
	h_{a,a'}^2(m,W)
	\right].
	\label{eq:continuous-C}
\end{equation}
The mediator-related quantities in \eqref{eq:continuous-C} are determined by
$(q_0,q_1,R_M^2)$, and $\gamma_a$ is calibrated from a target NIE $\delta(a)$ through $\gamma_a=\delta(a)/(q_1-q_0)$; a nonzero target NIE requires $q_1\neq q_0$, so the ratio is well defined. These inputs and $S_a^2$ do not determine $b_a$, so evaluating
$C_{a,a'}$ requires one additional input. We introduce $\kappa_a = \operatorname{Corr}\{Y(a,M(a)),W\}$, the correlation between the potential outcome and the baseline-covariate summary. Because $M(a)$ and $W$ are associated, translating $\kappa_a$ into $b_a$
also requires the mediator-score covariance $\xi_a=\operatorname{Cov}\{M(a),W\}$. Using the Gaussian marginal distribution of $W$, $\xi_a=\operatorname{Cov}\{M(a),W\}
	=
	\operatorname{Cov}\{p_a^W(W),W\} = \frac{R_M^2}{\sqrt{1-R_M^2}}
	\phi\{\Phi^{-1}(q_a)\},
$
where $\phi$ is the standard normal density, and the working outcome model gives $
	\operatorname{Cov}\{Y(a,M(a)),W\}
	=
	b_a\sigma_W^2+\gamma_a	\xi_a.
$
For $R_M^2>0$, equating this expression with $\kappa_aS_a\sigma_W$, where
$S_a=(S_a^2)^{1/2}$, gives
\begin{equation*}
	b_a
	=
	\frac{
		\kappa_aS_a\sigma_W-\gamma_a\xi_a
	}{
		\sigma_W^2
	}.
\end{equation*}
Substituting this expression for $b_a$ into $h_{a,a'}$ in
\eqref{eq:continuous-C} leaves a one-dimensional expectation over $W$.
When $R_M^2=0$, $W$ is degenerate, so $\kappa_a$ and the projection term
are omitted and $h_{a,a'}(m)=\gamma_a(m-q_{a'})$.
%Derivations of these equalities are provided in
%Section~\ref{app:continuous-outcome-working-derivation} of the
%Supplementary Material.

\textit{Evaluation of $K_{a}$.} Because $R_{a,a}(M,\mathbf X)\equiv 1$, the definition of $K_a$ in Section \ref{sec:variance-decomposition} simplifies to
\[
K_a
=
\operatorname{Cov}\left[
R_{a,1-a}(M,\mathbf X)\{Y-\theta(a,1-a)\},
Y-\theta(a,a)
\mid A=a
\right].
\]
Under model \eqref{eq:continuous-outcome-projection}, \begin{equation*}
	K_a
	=
	E_W\left[
	\sum_{m=0}^1
	g_{1-a}^W(m\mid W)
	\left\{
	h^2_{a,1-a}(m,W)
	+
	\sigma_{v,a}^2
	\right\}\right] \ge 0.
\end{equation*}
The residual variance is determined
by the planning inputs through \begin{equation*}
	\sigma_{v,a}^2
	=
	S_a^2
	-
	b_a^2\sigma_W^2
	-
	\gamma_a^2q_a(1-q_a)
	-
	2b_a\gamma_a\xi_a.
\end{equation*}
The requirement $\sigma_{v,a}^2\geq0$ is a coherence check: a negative value indicates that the specified inputs cannot simultaneously hold
under the working model.
When $R_M^2=0$, this expression reduces to $\sigma_{v,a}^2
=
S_a^2-\gamma_a^2q_a(1-q_a).$

\textit{Design-stage calculation.} The preceding calculations determine every component of the NDE and NIE
variance decompositions of Section~\ref{sec:variance-decomposition}.
Table~\ref{tab:continuous-outcome-planning-inputs} collects the results: given the design-stage inputs (Panel A), the calibration and intermediate working quantities (Panel B) determine the variance components (Panel C), which, substituted into the variance decompositions, yield the sample size through \eqref{eq:wald-sample-size}. The
tabulated expressions hold under Assumption~\ref{seq ign}, the
oracle-ratio moment condition
$E\{R_{a,a'}^2(M,\mathbf X)Y^2\mid A=a\}<\infty$, and the working models
of this subsection with the coherence condition $\sigma_{v,a}^2\geq0$.

\begin{table}[H]
\centering
\caption{Design-stage inputs, working quantities, and variance components
for a continuous outcome.}
\label{tab:continuous-outcome-planning-inputs}
\footnotesize
\setlength{\tabcolsep}{2pt}
\renewcommand{\arraystretch}{1.16}
\begin{tabular}{@{}>{\raggedright\arraybackslash}m{0.19\linewidth}
>{\centering\arraybackslash}m{0.12\linewidth}
>{\raggedright\arraybackslash}m{0.45\linewidth}
>{\centering\arraybackslash}m{0.16\linewidth}@{}}
\toprule
Quantity & Notation & Definition or working-model expression
& Design-stage inputs \\
\midrule
\multicolumn{4}{@{}l}{\textit{Panel A. Design-stage inputs}} \\
Treatment allocation
& $\pi_a$
& $P(A=a)$, $a=0,1$
& (specified) \\
Mediator prevalence
& $q_a$
& $P\{M(a)=1\}$, $a=0,1$
& (specified) \\
Latent mediator $R^2$
& $R_M^2$
& proportion of within-arm latent-index variance explained by $W$
& (specified) \\
Target NIE
& $\delta_*$
& target value of $\delta(a)$; NIE planning
& (specified) \\
Target NDE
& $\zeta_*$
& target value of $\zeta(a)$; NDE planning
& (specified) \\
Marginal outcome variance
& $S_a^2$
& $\operatorname{Var}\{Y(a,M(a))\}$
& (specified) \\
Outcome--score correlation
& $\kappa_a$
& $\operatorname{Corr}\{Y(a,M(a)),W\}$
& (specified) \\
\addlinespace
\multicolumn{4}{@{}l}{\textit{Panel B. Calibration and intermediate
working quantities}} \\
Latent-score variance
& $\sigma_W^2$
& $\displaystyle \frac{R_M^2}{1-R_M^2}$
& $R_M^2$ \\
Latent-score mean
& $\mu_W$
& $\displaystyle \frac{\Phi^{-1}(q_0)}{\sqrt{1-R_M^2}}$
& $q_0,R_M^2$ \\
Treatment shift
& $\Delta_M$
& $\displaystyle \frac{\Phi^{-1}(q_1)-\Phi^{-1}(q_0)}{\sqrt{1-R_M^2}}$
& $q_0,q_1,R_M^2$ \\
Mediator probability
& $p_a^W(w)$
& $\displaystyle \Phi(w+a\Delta_M)$
& $q_0,q_1,R_M^2$ \\
Mediator mass
& $g_a^W(m\mid w)$
& $\displaystyle \{p_a^W(w)\}^m\{1-p_a^W(w)\}^{1-m}$
& $q_0,q_1,R_M^2$ \\
Mediator ratio
& $R_{a,a'}^W(m,w)$
& $\displaystyle \frac{g_{a'}^W(m\mid w)}{g_a^W(m\mid w)}$
& $q_0,q_1,R_M^2$ \\
Pathway coefficient
& $\gamma_a$
& $\displaystyle \frac{\delta_*}{q_1-q_0}$
& $\delta_*,q_0,q_1$ \\
Mediator--score covariance
& $\xi_a$
& $\displaystyle \frac{R_M^2}{\sqrt{1-R_M^2}}\phi\{\Phi^{-1}(q_a)\}$
& $q_a,R_M^2$ \\
Outcome--score coefficient
& $b_a$
& $\displaystyle
\frac{\kappa_aS_a\sigma_W-\gamma_a\xi_a}{\sigma_W^2}$
& $\displaystyle
\begin{gathered}
q_a,R_M^2,S_a^2,\\[-2pt]
\kappa_a,\gamma_a
\end{gathered}$ \\
Residual variance
& $\sigma_{v,a}^2$
& $\displaystyle
S_a^2-b_a^2\sigma_W^2-\gamma_a^2q_a(1-q_a)-2b_a\gamma_a\xi_a$
& $\displaystyle
\begin{gathered}
q_a,R_M^2,S_a^2,\\[-2pt]
\kappa_a,\gamma_a
\end{gathered}$ \\
Centered mean function
& $h_{a,a'}(m,w)$
& $\displaystyle b_a(w-\mu_W)+\gamma_a(m-q_{a'})$
& $\displaystyle
\begin{gathered}
q_0,q_1,R_M^2,S_a^2,\\[-2pt]
\kappa_a,\gamma_a
\end{gathered}$ \\
\addlinespace
\multicolumn{4}{@{}l}{\textit{Panel C. Variance components}} \\
Ratio second moment
& $d_{a,a'}$
& $\displaystyle 1+E_W\left[
\frac{\{p_{a'}^W(W)-p_a^W(W)\}^2}
{p_a^W(W)\{1-p_a^W(W)\}}
\right]$
& $q_0,q_1,R_M^2$ \\
Outcome second moment
& $B_{a,a'}$
& $\displaystyle S_a^2+\gamma_a^2(q_a-q_{a'})^2$
& $q_0,q_1,S_a^2,\gamma_a$ \\
Ratio--outcome covariance
& $C_{a,a'}$
& \resizebox{\linewidth}{!}{$\displaystyle
E_W\!\left[\sum_{m=0}^1 g_a^W(m\mid W)
\left\{\{R_{a,a'}^W(m,W)\}^2-d_{a,a'}\right\}
h_{a,a'}^2(m,W)\right]$}
& $\displaystyle
\begin{gathered}
q_0,q_1,R_M^2,S_a^2,\\[-2pt]
\kappa_a,\gamma_a
\end{gathered}$ \\
Shared-arm covariance
& $K_a$
& $\displaystyle E_W\!\left[\sum_{m=0}^1g_{1-a}^W(m\mid W)
\{h_{a,1-a}^2(m,W)+\sigma_{v,a}^2\}\right]$
& $\displaystyle
\begin{gathered}
q_0,q_1,R_M^2,S_a^2,\\[-2pt]
\kappa_a,\gamma_a
\end{gathered}$ \\
\bottomrule
\end{tabular}
\end{table}

%In summary, the design-stage calculation for a continuous outcome
%proceeds from Table~\ref{tab:continuous-outcome-planning-inputs}: given
%the Panel A inputs, Panels B and C determine the working variance
%components, which, substituted into the variance decompositions of
%Section~\ref{sec:variance-decomposition}, yield the sample size through
%\eqref{eq:wald-sample-size}. The calculation rests on two working
%descriptions, under which the tabulated expressions are exact: a probit
%latent-index model for the binary mediator with a Gaussian approximation
%to the baseline-covariate summary $W$, and an outcome model additive in
%the mediator within each treatment arm, with joint normality of
%$f_a(\mathbf X)$ and $W$ supporting the projection and conditional
%homoscedasticity determining the residual variance from $S_a^2$. The
%working model allows the prognostic component and the pathway
%coefficient to differ between arms but does not accommodate
%covariate-dependent mediator effects within an arm. The numerical study
%in Section~\ref{sec:numerical} additionally assesses the estimator when
%the mediator probabilities are estimated rather than known.

\subsection{Binary outcome}
The mediator working model and $d_{a,a'}$ are unchanged. Binary outcomes, however, have a mean-dependent variance and require a separate working model. We reevaluate $B_{a,a'}$, $C_{a,a'}$, and $K_a$ under a latent-index outcome model. For each treatment level $a$, consider \begin{equation*}
	Y(a,m)
	=
	\mathbb{I}\left\{
	f_a^*(\mathbf X)
	+
	\gamma_a^*(m-q_a)
	+
	\epsilon_a^*
	>
	0
	\right\},
	\qquad
	\epsilon_a^*\sim\mathcal{N}(0,1),
\end{equation*}
where $f_a^*(\mathbf X)$ is the prognostic component and
$\gamma_a^*$ is the mediator--outcome pathway strength, both on the latent scale, and the latent error $\epsilon_a^*$ is independent of $\{\mathbf X, M(a)\}$. As in the continuous-outcome case, for $R_M^2>0$, define
$\widetilde b_a
=
\frac{
	\operatorname{Cov}\{f_a^*(\mathbf X),W\}
}{
	\sigma_W^2
}$ and  write the population linear projection
\begin{equation*}
f_a^*(\mathbf X)
=
E\{f_a^*(\mathbf X)\}
+
\widetilde b_a(W-\mu_W)
+
v_a^*,
\end{equation*} with $v_a^*$
independent of $\{W, M(a)\}$ under the joint Gaussian working approximation for $\{f_a^*(\mathbf X), W\}$. The compound latent residual $u_a^*=v_a^*+\epsilon_a^*$ therefore satisfies 
$
	u_a^*\mid\{W,M(a)\}
	\sim
	\mathcal{N}(0,\sigma_{u,a}^2)$, where $\sigma_{u,a}^2=\operatorname{Var}(v_a^*)+1.$

Define the standardized coefficients
$\alpha_a^\rho=E\{f_a^*(\mathbf X)\}/\sigma_{u,a}$,
$b_a^\rho=\widetilde b_a/\sigma_{u,a}$, and
$\gamma_a^\rho=\gamma_a^*/\sigma_{u,a}$.
Under consistency, the conditional distribution of $Y$ given
$(A=a,M=m,W=w)$ is Bernoulli with event probability
\[
\rho_a(w,m)
=
\Phi\left\{
\alpha_a^\rho
+
b_a^\rho(w-\mu_W)
+
\gamma_a^\rho(m-q_a)
\right\}.
\]
The corresponding mediation functional is
\begin{equation}
	\theta(a,a')
	=
	E_W\left[
	\sum_{m=0}^1
	g_{a'}^W(m\mid W)\rho_a(W,m)
	\right].
\label{eq:binary-outcome-theta}
\end{equation}
Given the mediator inputs $(q_0,q_1,R_M^2)$, the three coefficients
$(\alpha_a^\rho,b_a^\rho,\gamma_a^\rho)$ in $\rho_a(w,m)$ are jointly
calibrated to match the marginal event probability $\theta(a,a)$, the
outcome--$W$ correlation $\kappa_a$, and the target NIE $\delta(a)$,
whenever the corresponding calibration equations admit a compatible
solution.  %The calibration equations and computational details are provided in Section~\ref{app:proof-binary-outcome-working} of the Supplementary Material.

\textit{Evaluation of $B_{a,a'}$.}
Because $Y$ is binary in arm $a$, its variance is
$\theta(a,a)\{1-\theta(a,a)\}$, so that
\begin{equation}
	B_{a,a'}
	=
	\theta(a,a)\{1-\theta(a,a)\}
	+
	\{\theta(a,a)-\theta(a,a')\}^2.
	\label{eq:binary-B}
\end{equation}
In particular,
$B_{a,a}
	=
	\theta(a,a)\{1-\theta(a,a)\}$ and $
	B_{a,1-a}
	=
	\theta(a,a)\{1-\theta(a,a)\}+\delta^2(a).
$
Thus $B_{a,a'}$ is determined by the marginal event probability $\theta(a,a)$ and the target NIE $\delta(a)$.

\textit{Evaluation of $C_{a,a'}$.}
Using the Bernoulli structure of the outcome and the centering of the
squared mediator ratio around $d_{a,a'}$, $C_{a,a'}$ can be reformulated as
\begin{equation}
	\begin{split}
		C_{a,a'}
		=
		\{1-2\theta(a,a')\}
		E_W\Bigg[
		\sum_{m=0}^1
		&g_a^W(m\mid W)
		\left\{
		\{R_{a,a'}^W(m,W)\}^2-d_{a,a'}
		\right\}
		\rho_a(W,m)
		\Bigg].
	\end{split}
	\label{eq:binary-C}
\end{equation}
Unlike the constant
residual variance in the continuous-outcome calculation, the Bernoulli
variance changes with $(W,M)$ and is incorporated through $\rho_a(W,m)$ in
\eqref{eq:binary-C}. When $a'=a$, the mediator ratio is identically one and
 $C_{a,a}=0$. For $a'\neq a$, Equation~\eqref{eq:binary-C} is a
one-dimensional expectation determined by the mediator inputs
$(q_0,q_1,R_M^2)$ and the outcome inputs
$\{\theta(a,a),\kappa_a,\delta(a)\}$.

\textit{Evaluation of $K_a$.}
Because $R_{a,a}^W\equiv1$, only the reweighting by $R_{a,1-a}^W$ enters the
definition of $K_a$ in Section~\ref{sec:variance-decomposition}. Under the
change of mediator distribution induced by this ratio, $Y$ remains binary
with event probability $\theta(a,1-a)$, and $K_a$ reduces to the variance of
this reweighted outcome:
\begin{equation}
	K_a
	=
	\theta(a,1-a)\{1-\theta(a,1-a)\}
	\geq0;
	\label{eq:binary-K}
\end{equation}
$K_a$ requires only $\theta(a,1-a)$, which is determined by the design inputs $\theta(a,a)$ and the target NIE $\delta(a)$.

\textit{Design-stage calculation} The preceding calculations determine every component of the variance
decompositions of Section~\ref{sec:variance-decomposition}.
Table~\ref{tab:binary-outcome-planning-inputs} collects the results: given the design-stage inputs (Panel A), the working quantities (Panel B) determine the variance components (Panel C), and the sample size follows from \eqref{eq:wald-sample-size} as in the continuous case. Because the marginal event probabilities $\theta(1,1)$ and $\theta(0,0)$ are already required as inputs to determine the outcome variances, they and the target NIEs together determine the NDE targets through $\zeta(a) = \theta(1,1)-\theta(0,0)-\delta(1-a)$, so the target NDE is not a separate design input here. The tabulated expressions hold under Assumption~\ref{seq ign}, the moment condition of Proposition~\ref{prop:oracle-rmpw-var}, and the working models of this subsection, whenever the calibration system admits a solution compatible with the specified inputs.

\begin{table}[H]
\centering
\caption{Design-stage inputs, working quantities, and variance components
for a binary outcome.}
\label{tab:binary-outcome-planning-inputs}
\footnotesize
\setlength{\tabcolsep}{2pt}
\renewcommand{\arraystretch}{1.16}
\begin{tabular}{@{}>{\raggedright\arraybackslash}m{0.19\linewidth}
>{\centering\arraybackslash}m{0.12\linewidth}
>{\raggedright\arraybackslash}m{0.45\linewidth}
>{\centering\arraybackslash}m{0.16\linewidth}@{}}
\toprule
Quantity & Notation & Definition or working-model expression
& Design-stage inputs \\
\midrule
\multicolumn{4}{@{}l}{\textit{Panel A. Design-stage inputs}} \\
Treatment allocation
& $\pi_a$
& $P(A=a)$, $a=0,1$
& (specified) \\
Mediator prevalence
& $q_a$
& $P\{M(a)=1\}$, $a=0,1$
& (specified) \\
Latent mediator $R^2$
& $R_M^2$
& proportion of within-arm latent-index variance explained by $W$
& (specified) \\
Target NIE
& $\delta_*$
& target value of $\delta(a)$; NIE planning
& (specified) \\
Marginal event probability
& $\theta(a,a)$
& $P\{Y(a,M(a))=1\}$
& (specified) \\
Outcome--score correlation
& $\kappa_a$
& $\operatorname{Corr}\{Y(a,M(a)),W\}$
& (specified) \\
\addlinespace
\multicolumn{4}{@{}l}{\textit{Panel B. Calibration and intermediate
working quantities}} \\
Latent-score variance
& $\sigma_W^2$
& $\displaystyle \frac{R_M^2}{1-R_M^2}$
& $R_M^2$ \\
Latent-score mean
& $\mu_W$
& $\displaystyle \frac{\Phi^{-1}(q_0)}{\sqrt{1-R_M^2}}$
& $q_0,R_M^2$ \\
Treatment shift
& $\Delta_M$
& $\displaystyle \frac{\Phi^{-1}(q_1)-\Phi^{-1}(q_0)}{\sqrt{1-R_M^2}}$
& $q_0,q_1,R_M^2$ \\
Mediator probability
& $p_a^W(w)$
& $\displaystyle \Phi(w+a\Delta_M)$
& $q_0,q_1,R_M^2$ \\
Mediator mass
& $g_a^W(m\mid w)$
& $\displaystyle \{p_a^W(w)\}^m\{1-p_a^W(w)\}^{1-m}$
& $q_0,q_1,R_M^2$ \\
Mediator ratio
& $R_{a,a'}^W(m,w)$
& $\displaystyle \frac{g_{a'}^W(m\mid w)}{g_a^W(m\mid w)}$
& $q_0,q_1,R_M^2$ \\
Probit coefficients
& $\alpha_a^\rho,b_a^\rho,\gamma_a^\rho$
& solved jointly from the calibration system in
Section~\ref{app:proof-binary-outcome-working} of the Supplementary Material
& $\displaystyle
\begin{gathered}
q_0,q_1,R_M^2,\theta(a,a),\\[-2pt]
\kappa_a,\delta_*
\end{gathered}$ \\
Outcome event probability
& $\rho_a(w,m)$
& $\displaystyle \Phi\{\alpha_a^\rho+b_a^\rho(w-\mu_W)
+\gamma_a^\rho(m-q_a)\}$
& $\displaystyle
\begin{gathered}
q_0,q_1,R_M^2,\theta(a,a),\\[-2pt]
\kappa_a,\delta_*
\end{gathered}$ \\
Mediation functional
& $\theta(a,a')$
& $\displaystyle E_W\!\left[\sum_{m=0}^1
g_{a'}^W(m\mid W)\rho_a(W,m)\right]$
& $\displaystyle
\begin{gathered}
q_0,q_1,R_M^2,\theta(a,a),\\[-2pt]
\kappa_a,\delta_*
\end{gathered}$ \\
Cross-arm functional
& $\theta(a,1-a)$
& $\displaystyle \theta(a,a)+(1-2a)\delta_*$
& $\theta(a,a),\delta_*$ \\
\addlinespace
\multicolumn{4}{@{}l}{\textit{Panel C. Variance components}} \\
Ratio second moment
& $d_{a,a'}$
& $\displaystyle 1+E_W\left[
\frac{\{p_{a'}^W(W)-p_a^W(W)\}^2}
{p_a^W(W)\{1-p_a^W(W)\}}
\right]$
& $q_0,q_1,R_M^2$ \\
Outcome second moment
& $B_{a,a'}$
& $\displaystyle \theta(a,a)\{1-\theta(a,a)\}
+\{\theta(a,a)-\theta(a,a')\}^2$
& $\theta(a,a),\delta_*$ \\
Ratio--outcome covariance
& $C_{a,a'}$
& \resizebox{\linewidth}{!}{$\displaystyle \{1-2\theta(a,a')\}E_W\!\left[
\sum_{m=0}^1g_a^W(m\mid W)
\left\{\{R_{a,a'}^W(m,W)\}^2-d_{a,a'}\right\}
\rho_a(W,m)\right]$}
& $\displaystyle
\begin{gathered}
q_0,q_1,R_M^2,\theta(a,a),\\[-2pt]
\kappa_a,\delta_*
\end{gathered}$ \\
Shared-arm covariance
& $K_a$
& $\displaystyle \theta(a,1-a)\{1-\theta(a,1-a)\}$
& $\theta(a,a),\delta_*$ \\
\bottomrule
\end{tabular}
\end{table}

\section{Numerical study} \label{sec:numerical}
We conduct simulations to assess whether the proposed design-stage sample
size calculation attains the nominal power for natural indirect and direct
effects, followed by a design-calibration example based on the ACTG175 trial.

All simulation scenarios use a randomized trial with fixed 1:1 treatment
allocation. For each estimand, its true value serves as the design alternative
$\tau_*(a)$ in Equation~\eqref{eq:wald-sample-size}, and the working-model
variance in Section~\ref{sec:binary-working-model} determines the analytic
sample size $N$ for 80\% power at a one-sided level $\alpha=0.05$. We estimate
power from 500 Monte Carlo replications at $N$ and obtain a simulation-based
benchmark $N_{\mathrm{sim}}$ by interpolating the simulated power curve.
Accordingly, $N/N_{\mathrm{sim}}$ close to one and empirical power close to
0.80 indicate good calibration.

Let $\mathbf X\in\mathbb R^{10}$ denote the baseline covariates. We consider
 two covariate designs with mutually independent components. In the Gaussian design, the components of $\mathbf X$
are standard normal variables. In the mixed design, $\mathbf X$
contains standardized Bernoulli, uniform, Poisson, gamma, and beta
components.
The binary mediator is generated from the probit model
\[
P(M=1\mid A=a,\mathbf X)
=
p_a(\mathbf X)
=
\Phi(W+a\Delta_M),
\qquad
W=\mu_W+\mathbf X^\top\beta_M,
\]
so that $W$ is the baseline-covariate summary of
Section~\ref{sec:binary-working-model}. The intercept $\mu_W$ and the
treatment shift $\Delta_M$ are calibrated so that $E\{p_0(\mathbf X)\}=q_0$
and $E\{p_1(\mathbf X)\}=q_1$, and the coefficient vector $\beta_M$ is
scaled so that
$
R_M^2
=
\frac{\operatorname{Var}(\mathbf X^\top\beta_M)}
{1+\operatorname{Var}(\mathbf X^\top\beta_M)}
$
attains its design value. Complete distributional specifications are provided in
Section~\ref{app:additional-oracle-results} of the Supplementary Material.
In the Gaussian design, $W$ is exactly normal, so
the mediator working model holds exactly; the mixed design departs from
normality and probes the Gaussian approximation to $W$. 
    
\subsection{Continuous outcome}
The continuous outcome is generated as
\[
Y=\eta_A+f_Y(\mathbf X)+\gamma_AM+\epsilon_A,
\qquad
f_Y(\mathbf X)=\mathbf X^\top\beta_Y,
\qquad
\epsilon_A\sim\mathcal{N}(0,\sigma_A^2).
\]
The coefficient vectors $\beta_M$ and $\beta_Y$ are chosen so that
$\operatorname{Corr}\{W,f_Y(\mathbf X)\}=\rho$. The design therefore allows
the baseline covariates to jointly predict the mediator and the outcome,
and it allows treatment--mediator interaction through
$\gamma_0\neq\gamma_1$. Under this data-generating mechanism,
$
\theta(a,a')
=
E\{\eta_a+f_Y(\mathbf X)+\gamma_ap_{a'}(\mathbf X)\}
=
\eta_a+\gamma_aq_{a'},
$
so the natural effects are $\delta(a)=\gamma_a(q_1-q_0)$ and
$\zeta(a)=\eta_1-\eta_0+(\gamma_1-\gamma_0)q_a$. We evaluate the four
natural-effect estimands $\delta(1)$, $\delta(0)$, $\zeta(1)$, and
$\zeta(0)$.

The baseline setting uses $\pi_1=0.5$, $q_0=0.35$, $q_1=0.55$, $R_M^2=0.25$,
$\rho=0.5$, $\eta_0=0$, $\eta_1-\eta_0=0.08$, $\gamma_0=\gamma_1=0.4$, and $\sigma_0=\sigma_1=1$. We consider a range of
simulation panels, each varying one design feature around this baseline.
Complete panel definitions and the corresponding mixed-covariate results
for both outcome types are provided in
Section~\ref{app:additional-oracle-results} of the Supplementary Material.

Table~\ref{tab:continuous-gaussian} reports the oracle-ratio results under
Gaussian covariates: for each estimand and panel setting, the analytic
$N$, the benchmark $N_{\mathrm{sim}}$, and the empirical power at $N$, with
both $N$ and $N_{\mathrm{sim}}$ computed using the oracle mediator
probability ratios. The analytic calculation is well calibrated across the continuous-outcome
settings. The results below are computed with the oracle mediator
probability ratios; the estimated-ratio analysis is reported afterward.
Under Gaussian covariates (Table~\ref{tab:continuous-gaussian}), the mean
empirical power at $N$ is 0.799, with values ranging from 0.766 to 0.846
across panels and estimands, and the average ratio $N/N_{\mathrm{sim}}$ is
0.997. Under mixed covariates (Table~\ref{tab:continuous-mixed} of the
Supplementary Material), the corresponding figures are 0.804, 0.750 to
0.842, and 1.002. In both designs the analytic formula therefore produces
sample sizes close to those obtained by direct simulation, and the close
agreement between the two designs shows that the accuracy of the
calculation does not depend on Gaussian covariates. 

\begin{table}[H]
\centering
\scriptsize
\caption{Oracle-ratio continuous-outcome simulations with Gaussian covariates.}
\label{tab:continuous-gaussian}
\setlength{\tabcolsep}{2pt}
\resizebox{\textwidth}{!}{%
\begin{tabular}{l*{12}{c}}
\simtableheader
Baseline & \simresult{437}{455}{0.786} & \simresult{477}{467}{0.802} & \simresult{5246}{5506}{0.794} & \simresult{5166}{5610}{0.790} \\
$\gamma_0=\gamma_1=0.2$ & \simresult{1757}{1786}{0.810} & \simresult{1916}{1916}{0.800} & \simresult{5066}{5116}{0.806} & \simresult{5013}{5123}{0.784} \\
$\gamma_0=\gamma_1=0.4$ & \simresult{436}{425}{0.804} & \simresult{476}{482}{0.794} & \simresult{5235}{5385}{0.772} & \simresult{5155}{5560}{0.796} \\
$\gamma_0=\gamma_1=0.6$ & \simresult{194}{186}{0.814} & \simresult{212}{214}{0.798} & \simresult{5514}{5622}{0.780} & \simresult{5395}{5395}{0.800} \\
$\eta_1-\eta_0=0.04$ & \simresult{436}{457}{0.774} & \simresult{475}{504}{0.790} & \simresult{20963}{20072}{0.794} & \simresult{20644}{20037}{0.794} \\
$\eta_1-\eta_0=0.08$ & \simresult{438}{439}{0.834} & \simresult{477}{486}{0.792} & \simresult{5246}{5222}{0.806} & \simresult{5166}{5246}{0.786} \\
$\eta_1-\eta_0=0.12$ & \simresult{437}{448}{0.794} & \simresult{476}{434}{0.818} & \simresult{2331}{2423}{0.784} & \simresult{2296}{2450}{0.788} \\
$\{q_0,q_1\}=\{0.40,0.60\}$ & \simresult{449}{419}{0.846} & \simresult{449}{425}{0.822} & \simresult{5195}{5084}{0.808} & \simresult{5195}{5077}{0.808} \\
$\{q_0,q_1\}=\{0.20,0.40\}$ & \simresult{462}{468}{0.796} & \simresult{701}{743}{0.780} & \simresult{5497}{5073}{0.808} & \simresult{5092}{4794}{0.832} \\
$\{q_0,q_1\}=\{0.60,0.80\}$ & \simresult{702}{672}{0.824} & \simresult{462}{437}{0.822} & \simresult{5086}{4930}{0.810} & \simresult{5491}{6011}{0.770} \\
$R_M^2=0$ & \simresult{363}{385}{0.772} & \simresult{394}{401}{0.794} & \simresult{5028}{4847}{0.798} & \simresult{4971}{4788}{0.818} \\
$R_M^2=0.25$ & \simresult{436}{431}{0.806} & \simresult{475}{503}{0.796} & \simresult{5240}{5241}{0.796} & \simresult{5160}{5221}{0.782} \\
$R_M^2=0.50$ & \simresult{586}{572}{0.810} & \simresult{642}{609}{0.826} & \simresult{5461}{5408}{0.766} & \simresult{5358}{5310}{0.804} \\
$\rho=0$ & \simresult{449}{444}{0.786} & \simresult{491}{488}{0.802} & \simresult{5129}{5081}{0.802} & \simresult{5062}{5099}{0.782} \\
$\rho=0.5$ & \simresult{436}{416}{0.826} & \simresult{475}{486}{0.788} & \simresult{5236}{5589}{0.772} & \simresult{5156}{5134}{0.802} \\
$\rho=0.8$ & \simresult{428}{433}{0.796} & \simresult{465}{479}{0.796} & \simresult{5309}{5239}{0.824} & \simresult{5222}{5137}{0.806} \\
$\{\gamma_0,\gamma_1\}=\{0.30,0.50\}$ & \simresult{279}{283}{0.796} & \simresult{848}{874}{0.810} & \simresult{933}{938}{0.798} & \simresult{1474}{1421}{0.826} \\
$\{\gamma_0,\gamma_1\}=\{0.20,0.60\}$ & \simresult{194}{203}{0.782} & \simresult{1915}{2053}{0.798} & \simresult{375}{378}{0.798} & \simresult{685}{734}{0.770} \\
\bottomrule
\end{tabular}
}
\end{table}

We perform two additional checks. First, we estimate the mediator probability ratios within each simulated trial using a probit model. Because the analytic variance does not include the first-order contribution of nuisance estimation, this estimated-ratio analysis serves as an empirical robustness check.
It gives a mean empirical power of 0.799 (range 0.748--0.846), pooled across covariate designs, panels, and estimands. Second, under null NIE and NDE scenarios, the type-I errors are 0.056 and 0.050 using oracle ratios and 0.051 and 0.046 using estimated ratios, all close to the nominal 0.05 level. Detailed results are reported in Section~\ref{app:estimated-ratio-results} of the Supplementary Material.
\subsection{Binary outcome}
The binary outcome is
generated from the probit model
\[
\begin{aligned}
P(Y=1\mid A=a,M=m,\mathbf{X})
&=
\mu_a(m,\mathbf{X})
=
\Phi\{\alpha_Y+\eta_a+f_Y(\mathbf{X})+\gamma_a m\}, 
\end{aligned}
\]
$f_Y(\mathbf{X})=\mathbf{X}^\top\beta_Y .$
The intercept $\alpha_Y$ is calibrated so that 
$\theta(0,0)$ equals the target risk $r_0$. Under this data-generating mechanism,
$
\theta(a,a')
=
E\left[
\mu_a(1,\mathbf{X})p_{a'}(\mathbf{X})+\mu_a(0,\mathbf{X})\{1-p_{a'}(\mathbf{X})\}
\right],
$
where $p_{a'}(\mathbf{X})=P(M=1\mid A=a',\mathbf{X})$. As in the continuous-outcome
study, we evaluate $\delta(1)$, $\delta(0)$, $\zeta(1)$, and $\zeta(0)$.
The baseline setting uses $\pi_1=0.5$, $q_0=0.35$, $q_1=0.55$,
$R_M^2=0.25$, $\rho=0.5$, $\eta_0=0$, $\eta_1=0.25$, $\gamma_0=\gamma_1=0.35$, and $r_0=0.30$. The simulation panels vary one design feature around this baseline as in the continuous case, panel definitions
are given in Section~\ref{app:additional-oracle-results} of the
Supplementary Material.

Table~\ref{tab:binary-gaussian} reports the oracle-ratio results under
Gaussian covariates: for each estimand and panel setting, the analytic
$N$, the benchmark $N_{\mathrm{sim}}$, and the empirical power at $N$, with
both $N$ and $N_{\mathrm{sim}}$ computed using the oracle mediator
probability ratios. The analytic calculation is also well calibrated for
binary outcomes. Under Gaussian covariates, the mean empirical power at
$N$ is 0.804, with values ranging from 0.764 to 0.844 across panels and
estimands, and the average ratio $N/N_{\mathrm{sim}}$ is 1.019. Under
mixed covariates (Table~\ref{tab:binary-mixed} of the Supplementary
Material), the corresponding figures are 0.798, 0.748 to 0.848, and 1.008.
The analytic sample sizes are thus close to the direct simulation
benchmarks, and the close agreement again shows
that the accuracy of the calculation does not depend on Gaussian
covariates.

\begin{table}[H]
\centering
\scriptsize
\caption{Oracle-ratio binary-outcome simulations with Gaussian covariates.}
\label{tab:binary-gaussian}
\setlength{\tabcolsep}{2pt}
\resizebox{\textwidth}{!}{%
\begin{tabular}{l*{12}{c}}
\simtableheader
Baseline & \simresult{923}{909}{0.808} & \simresult{1087}{1094}{0.786} & \simresult{885}{883}{0.808} & \simresult{879}{867}{0.810} \\
$\gamma_0=\gamma_1=0.20$ & \simresult{2851}{2968}{0.786} & \simresult{3302}{3085}{0.796} & \simresult{861}{841}{0.824} & \simresult{859}{823}{0.812} \\
$\gamma_0=\gamma_1=0.35$ & \simresult{920}{913}{0.796} & \simresult{1085}{1005}{0.830} & \simresult{883}{891}{0.796} & \simresult{877}{882}{0.784} \\
$\gamma_0=\gamma_1=0.50$ & \simresult{449}{445}{0.812} & \simresult{538}{538}{0.800} & \simresult{920}{871}{0.800} & \simresult{909}{870}{0.810} \\
$\eta_1-\eta_0=0.15$ & \simresult{938}{985}{0.780} & \simresult{1084}{1088}{0.828} & \simresult{2495}{2320}{0.766} & \simresult{2466}{2349}{0.818} \\
$\eta_1-\eta_0=0.25$ & \simresult{924}{885}{0.802} & \simresult{1088}{1041}{0.794} & \simresult{885}{844}{0.804} & \simresult{879}{827}{0.828} \\
$\eta_1-\eta_0=0.35$ & \simresult{910}{906}{0.808} & \simresult{1086}{1077}{0.804} & \simresult{445}{417}{0.796} & \simresult{443}{426}{0.820} \\
$\{q_0,q_1\}=\{0.40,0.60\}$ & \simresult{952}{1037}{0.786} & \simresult{1031}{1036}{0.810} & \simresult{876}{787}{0.818} & \simresult{881}{799}{0.816} \\
$\{q_0,q_1\}=\{0.20,0.40\}$ & \simresult{952}{941}{0.824} & \simresult{1549}{1638}{0.772} & \simresult{930}{905}{0.800} & \simresult{874}{857}{0.812} \\
$\{q_0,q_1\}=\{0.60,0.80\}$ & \simresult{1478}{1434}{0.814} & \simresult{1056}{1045}{0.832} & \simresult{855}{845}{0.816} & \simresult{917}{939}{0.798} \\
$R_M^2=0$ & \simresult{750}{723}{0.822} & \simresult{901}{874}{0.786} & \simresult{848}{856}{0.796} & \simresult{840}{841}{0.794} \\
$R_M^2=0.25$ & \simresult{920}{995}{0.772} & \simresult{1084}{1012}{0.800} & \simresult{884}{846}{0.780} & \simresult{878}{838}{0.804} \\
$R_M^2=0.50$ & \simresult{1234}{1092}{0.806} & \simresult{1461}{1460}{0.816} & \simresult{923}{895}{0.802} & \simresult{917}{895}{0.818} \\
$\rho=0$ & \simresult{923}{898}{0.806} & \simresult{1124}{1038}{0.844} & \simresult{865}{880}{0.770} & \simresult{856}{842}{0.806} \\
$\rho=0.5$ & \simresult{920}{875}{0.828} & \simresult{1085}{1075}{0.816} & \simresult{884}{910}{0.780} & \simresult{877}{884}{0.764} \\
$\rho=0.8$ & \simresult{917}{895}{0.818} & \simresult{1059}{1027}{0.786} & \simresult{897}{897}{0.800} & \simresult{892}{819}{0.792} \\
$\{\gamma_0,\gamma_1\}=\{0.25,0.45\}$ & \simresult{549}{578}{0.818} & \simresult{2118}{2132}{0.768} & \simresult{403}{406}{0.776} & \simresult{510}{486}{0.830} \\
$\{\gamma_0,\gamma_1\}=\{0.15,0.55\}$ & \simresult{364}{377}{0.788} & \simresult{5854}{6142}{0.802} & \simresult{229}{218}{0.832} & \simresult{335}{332}{0.814} \\
$r_0=0.15$ & \simresult{1122}{1117}{0.802} & \simresult{1477}{1436}{0.792} & \simresult{1122}{1116}{0.804} & \simresult{1123}{1109}{0.820} \\
$r_0=0.30$ & \simresult{923}{930}{0.798} & \simresult{1087}{1076}{0.822} & \simresult{885}{833}{0.786} & \simresult{879}{853}{0.806} \\
$r_0=0.50$ & \simresult{935}{905}{0.812} & \simresult{987}{1046}{0.790} & \simresult{849}{824}{0.808} & \simresult{836}{819}{0.834} \\
\bottomrule
\end{tabular}
}
\end{table}

The two additional checks give the same picture as in the continuous
outcome. Estimating the mediator probability ratios within each simulated
trial gives mean empirical power 0.800 (range 0.758 to 0.850) pooled across
covariate designs, panels, and estimands, nearly identical to the
oracle-ratio analysis. In null scenarios, the type-I errors of the Wald
test are 0.051 for the NIE and 0.052 for the NDE in the oracle-ratio
analysis, and 0.046 and 0.050 in the estimated-ratio analysis, all close
to the nominal 0.05 level. Detailed binary-outcome results are reported in
Section~\ref{app:estimated-ratio-results} of the Supplementary Material.

As a complementary comparison, we evaluate the power of the natural-effect
tests in trials sized for the ATE. In three representative continuous-outcome and three binary-outcome Gaussian-covariates settings, we compute the two-arm ATE sample size $N_{\mathrm{ATE}}$ for 80\% power and evaluate the tests in 1000 simulated trials of that size. Empirical ATE power ranges from 0.790 to 0.816, whereas power ranges from 0.249 to 1.000 for the NIE and from 0.199 to 0.655 for the NDE. Because ATE-based planning does not incorporate the mediator pathway, $N_{\mathrm{ATE}}$ may be smaller or larger than the sample size required for a prespecified natural effect. The complete comparison is reported in Section~\ref{app:ate-comparison} of the Supplementary Material.

\subsection{ACTG175 design-calibration example}
\label{sec:actg175}
We illustrate how pilot trial data can calibrate the design inputs of the
proposed calculation using the ACTG175 randomized trial, which compared zidovudine monotherapy with other antiretroviral therapies in
adults with HIV infection and baseline CD4 T-cell counts between 200 and 500
cells per cubic millimeter \citep{hammer1996trial}. The data are available through the \texttt{speff2trial} package \citep{juraska2022speff2trial}.

We define $A=0$ as zidovudine monotherapy and $A=1$ as the pooled other therapy arms. The primary mediator is the early immunologic response
$M=I(\Delta\mathrm{CD4}\ge50)$, where $\Delta\mathrm{CD4}$ is the change in
CD4 T-cell count from baseline to week 20. The 50-cell threshold is consistent
with previously used definitions of short-term immunologic response
\citep{bisson2006diagnostic,moore2009effect,boatman2019risk}; thresholds of 25 and 100 cells are considered in sensitivity
analyses. The outcome is an indicator of remaining free through week 96 of
CD4 decline of at least 50 cells, progression to AIDS, or death.
{We restrict the analysis to participants with observed mediator and outcome,
yielding a working pilot population of 1941 participants ($468$ in arm $0$ and
$1473$ in arm $1$).} At the 25-,
50-, and 100-cell thresholds, mediator prevalences are 0.314, 0.231, and
0.103 in arm 0 and 0.506, 0.428, and 0.258 in arm 1; the corresponding
arm-specific 96-week event-free proportions are 0.752 and 0.870.
{For illustration, we do not address the missingness in the mediator and the outcome;} these quantities are used only as descriptive pilot-calibration
inputs for illustrating how the proposed formulas operate and should not be
interpreted as causal NIE or NDE estimates for the full randomized cohort.
%These
%estimates and design calibrations pertain to participants who are event-free
%and observed at the 20-week landmark rather than to the full randomized cohort.

For each threshold, we fit a probit model for $P(M=1\mid A,\mathbf X)$ using
baseline covariates and use the resulting mediator probability ratios in the RMPW estimators. At the corresponding design-stage sample size $N$, we estimate power from 500 arm-stratified bootstrap resamples that preserve the observed allocation ratio. For each resample, we recompute the RMPW point estimate and apply the one-sided Wald-type test using the
corresponding design-stage variance. Additional details on the construction of the pilot population,
calibration of the design inputs, and resampling procedure are
provided in Section~\ref{app:actg175-calibration} of the Supplementary Material.

\begin{table}[H]
	\centering
	\caption{ACTG175 pilot-calibrated design quantities and fixed-ratio
		resampling results by mediator threshold.}
	\label{tab:actg175-sample-size}
	
	\small
	\setlength{\tabcolsep}{4pt}
	\renewcommand{\arraystretch}{1.08}
	
	\begin{tabular*}{\textwidth}
		{@{\extracolsep{\fill}}llcccc@{}}
		\toprule
		Mediator definition
		& Estimand
		& Pilot target
		& $V$
		& $N$
		& Resampling power \\
		\midrule
		
		$\Delta\mathrm{CD4}\ge 25$
		& $\delta(1)$ & 0.0293 & 0.0251 & 181  & 0.812 (0.017) \\
		& $\delta(0)$ & 0.0373 & 0.1054 & 469  & 0.780 (0.019) \\
		& $\zeta(1)$  & 0.0809 & 0.8598 & 813  & 0.788 (0.018) \\
		& $\zeta(0)$  & 0.0889 & 1.0019 & 784  & 0.786 (0.018) \\
		
		\addlinespace[3pt]
		$\Delta\mathrm{CD4}\ge 50$
		& $\delta(1)$ & 0.0285 & 0.0229 & 174  & 0.790 (0.018) \\
		& $\delta(0)$ & 0.0333 & 0.1260 & 705  & 0.778 (0.019) \\
		& $\zeta(1)$  & 0.0849 & 0.8997 & 771  & 0.788 (0.018) \\
		& $\zeta(0)$  & 0.0897 & 0.9983 & 768  & 0.800 (0.018) \\
		
		\addlinespace[3pt]
		$\Delta\mathrm{CD4}\ge 100$
		& $\delta(1)$ & 0.0191 & 0.0121 & 207  & 0.780 (0.019) \\
		& $\delta(0)$ & 0.0259 & 0.1275 & 1174 & 0.816 (0.017) \\
		& $\zeta(1)$  & 0.0923 & 0.9354 & 680  & 0.814 (0.017) \\
		& $\zeta(0)$  & 0.0991 & 0.9703 & 611  & 0.788 (0.018) \\
		
		\bottomrule
	\end{tabular*}
	
	\vspace{3pt}
	\parbox{\textwidth}{\footnotesize
		Pilot target is the normalized RMPW estimate obtained from the full pilot
		sample. $V$ is the model-implied design-stage variance, and $N$ is the
		sample size for 80\% power under a one-sided level-0.05 test. Resampling power and its
		Monte Carlo standard error are based on 500 arm-preserving resamples at
		$N$. The 50-cell threshold is the primary
		mediator definition.
	}
	\label{main:lastpage}
\end{table}

Table~\ref{tab:actg175-sample-size} reports the pilot-calibrated targets, design-stage variances, sample sizes,
and resampling power. At the primary 50-cell
threshold,  the pilot targets are 0.029 for $\delta(1)$, 0.033 for
$\delta(0)$, 0.085 for $\zeta(1)$, and 0.090 for $\zeta(0)$. Despite the
similar NIE targets, $\delta(1)$ and $\delta(0)$ require markedly different
sample sizes, 174 and 705, because their design-stage variances are 0.0229
and 0.1260.  This contrast reflects differences in the effect-specific
treatment arm, mediator overlap, outcome variation, and shared-arm
covariance. At the 100-cell threshold, the required sample size for
$\delta(0)$ increases to 1174, whereas those for $\zeta(1)$ and $\zeta(0)$
decrease to 680 and 611. Across all 12 threshold--estimand combinations,
resampling power ranges from 0.778 to 0.816, close to the target of 0.80.

\begin{figure}[htbp]
\centering
\includegraphics[width=\textwidth]{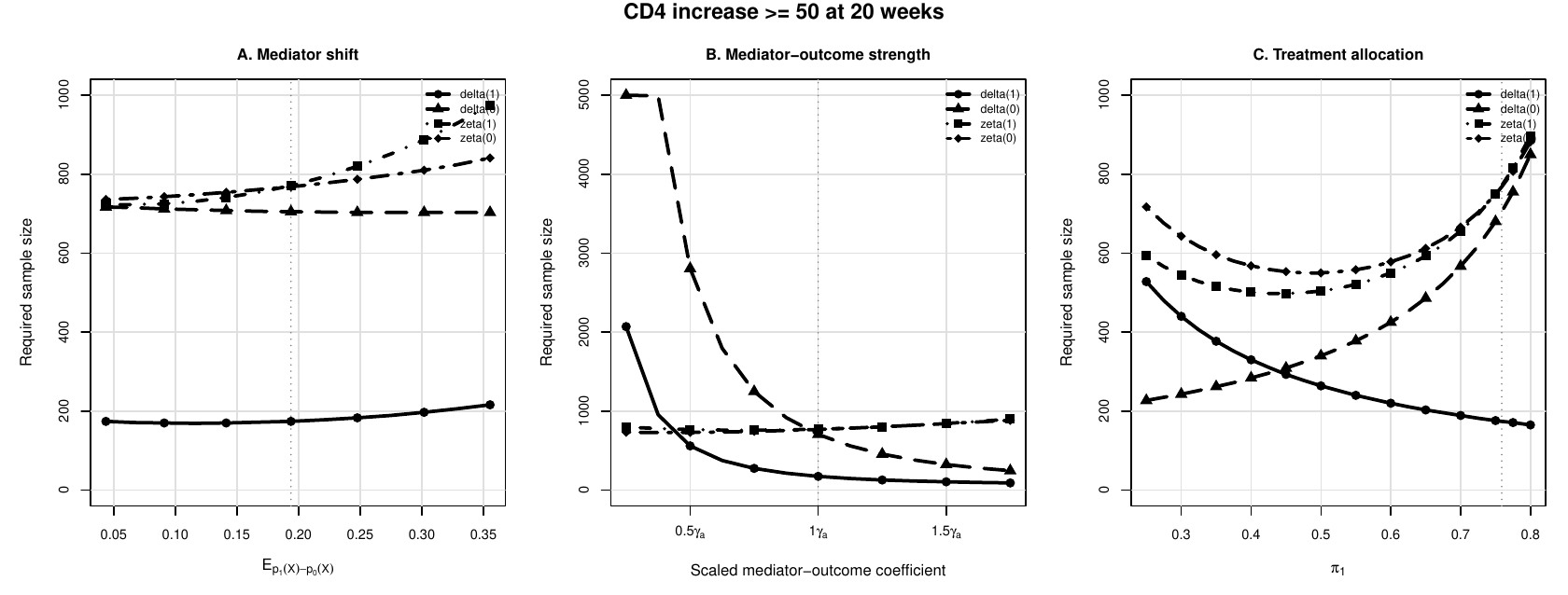}
\caption{ACTG175-calibrated working-design sensitivity for the primary mediator definition $M=1\{\Delta\mathrm{CD4}\ge 50\}$. The panels vary the latent-probit treatment-to-mediator shift, jointly scale the arm-specific mediator--outcome coefficients $\gamma_a^\rho$, and vary the treatment allocation while holding the remaining pilot-calibrated quantities fixed. Vertical dotted lines indicate the baseline values.}
\label{fig:actg175-design-sensitivity}
\end{figure}

Figure~\ref{fig:actg175-design-sensitivity} shows how the sample size required for each natural effect responds to changes in the design inputs. Starting from the working laws calibrated at the primary threshold, we vary one input at a time, holding the others fixed, and recompute the natural effects, design-stage variances, and required sample sizes.  The inputs varied
are the latent treatment--mediator shift $\Delta_M$, the mediator--outcome
pathway coefficients $\gamma_a^\rho$ scaled jointly across arms, and the
treatment allocation $\pi_1$. Each panel plots $N$ against one input, with
the pilot-calibrated value indicated by a vertical dotted line. The required
sample size is most sensitive to the mediator--outcome pathway: as
$\gamma_a^\rho$ decreases, the variance-to-squared-effect ratio increases,
leading to a particularly sharp increase in the sample sizes for the indirect
effects. Changes in $\Delta_M$ have more moderate, estimand-specific effects.
The allocation $\pi_1$ also produces different patterns because the four
estimators draw information differently from the two treatment arms; hence,
the allocation minimizing $N$ need not be the same across natural effects.

\section{Conclusion}\label{sec:conclusion}
This paper develops power and sample size calculation methods for causal
mediation analysis with a binary mediator in randomized trials, using the RMPW estimators. The variance decomposition of
Section~\ref{sec:variance-decomposition} makes explicit how the required
sample size depends on treatment allocation, mediator prevalence and
overlap, outcome variation, the mediator--outcome pathway, and the
outcome--covariate association, and the working models of
Section~\ref{sec:binary-working-model} reduce these to a small set of
design-stage inputs for continuous and binary outcomes. For a binary
outcome, the marginal event probabilities replace the outcome variances as
inputs, and the direct effects are then determined by the event
probabilities and the indirect effects. Simulation studies showed that the analytic calculations are well
calibrated for both outcome types and under both Gaussian and non-Gaussian
covariates, with similar empirical power under oracle and estimated
mediator probability ratios.
The ACTG175 example illustrated how pilot data can calibrate the design
inputs and showed that the required sample size, and the allocation that
minimizes it, differ across natural effects even at similar effect sizes. Future work could extend the design calculations to alternative causal
mediation estimands, such as interventional effects, and to settings with multivalued or longitudinal mediators.

\FloatBarrier

\bibliographystyle{plainnat}
\bibliography{bibl1.bib}

\clearpage
\section*{Supplementary Material for ``Power and sample size calculations for causal mediation analysis with a binary mediator in randomized trials"}
The Supplementary Material is organized as follows. Section~\ref{app:main-text-proofs}
provides proofs and derivations for the main-text results, including the
oracle RMPW variance decomposition, a Lyapunov central limit theorem, and
the continuous- and binary-outcome design-stage variance formulas.
Section~\ref{app:additional-oracle-results} describes the complete simulation
settings and reports additional oracle-ratio results.
Section~\ref{app:estimated-ratio-results} presents simulation results when the
mediator probability ratios are estimated within each simulated trial.
Section~\ref{app:ate-comparison} compares natural-effect-based planning with
conventional total-effect-based planning. Section~\ref{app:actg175-calibration} documents the construction of the
ACTG175 pilot population, the calibration of the design inputs,
and the resampling procedure used in the design illustration. Finally, Section~\ref{app:continuous-regression} develops the continuous-mediator
regression working model and provides the corresponding theorem and
corollary proofs.

\appendix
\setcounter{secnumdepth}{3}
\setcounter{section}{0}
\renewcommand{\thesection}{S\arabic{section}}
\renewcommand{\theHsection}{S\arabic{section}}
\renewcommand{\thesubsection}{\thesection.\arabic{subsection}}
\renewcommand{\thesubsubsection}{\thesubsection.\arabic{subsubsection}}
\renewcommand{\theequation}{\thesection.\arabic{equation}}
\renewcommand{\theHequation}{S\arabic{section}.\arabic{equation}}
\setcounter{table}{0}
\renewcommand{\thetable}{S\arabic{table}}
\renewcommand{\theHtable}{S\arabic{table}}
\setcounter{figure}{0}
\renewcommand{\thefigure}{S\arabic{figure}}
\renewcommand{\theHfigure}{S\arabic{figure}}
\setcounter{lemma}{0}
\renewcommand{\thelemma}{S\arabic{lemma}}
\renewcommand{\theHlemma}{S\arabic{lemma}}
\setcounter{theorem}{0}
\renewcommand{\thetheorem}{S\arabic{theorem}}
\renewcommand{\theHtheorem}{S\arabic{theorem}}
\setcounter{corollary}{0}
\renewcommand{\thecorollary}{S\arabic{corollary}}
\renewcommand{\theHcorollary}{S\arabic{corollary}}
\section{Proofs for main-text theoretical results}
\label{app:main-text-proofs}

This section collects an auxiliary normal-approximation result and proof details
deferred from the main text. The proof of the oracle variance decomposition is
followed by a Lyapunov central limit theorem that motivates the Gaussian
baseline-score working approximation, the derivation of the continuous-outcome and binary outcome
design-stage variances, in the
order in which these results appear in
Section~\ref{sec:binary-working-model}.

\subsection{Proof of the oracle RMPW variance decomposition}
\label{app:proof-oracle-rmpw-placeholder}

\begin{proof}[Proof of Proposition~\ref{prop:oracle-rmpw-var} and the
Section~\ref{sec:variance-decomposition} contrast formulas]
Let
\[
S_N(a,a')=N^{-1}\sum_{i=1}^N
I(A_i=a)R_{a,a'}(M_i,\mathbf X_i)Y_i
\]
and
\[
D_N(a,a')=N^{-1}\sum_{i=1}^N
I(A_i=a)R_{a,a'}(M_i,\mathbf X_i).
\]
By identification and the definition of the mediator probability ratio,
\[
E\{S_N(a,a')\}=\pi_a\theta(a,a'),
\qquad
E\{D_N(a,a')\}=\pi_a .
\]
A first-order Taylor expansion of $S_N(a,a')/D_N(a,a')$ around these
expectations gives
\[
\sqrt N\{\hat\theta(a,a')-\theta(a,a')\}
=
\frac{1}{\sqrt N}\sum_{i=1}^N
\frac{I(A_i=a)}{\pi_a}R_{a,a'}(M_i,\mathbf X_i)
\{Y_i-\theta(a,a')\}+o_p(1).
\]
The assumed second moment condition and the central limit theorem then yield
Proposition~\ref{prop:oracle-rmpw-var} and the variance
\eqref{eq:rmpw-theta-var}.

The design-effect decomposition \eqref{eq:theta-design-effect-decomp} follows
from
\[
E(R^2U^2\mid A=a)
=
E(R^2\mid A=a)E(U^2\mid A=a)
+
\operatorname{Cov}(R^2,U^2\mid A=a),
\]
with $R=R_{a,a'}(M,\mathbf X)$ and
$U=Y-\theta(a,a')$. The identity
$d_{a,a'}=1+\operatorname{Var}\{R_{a,a'}(M,\mathbf X)\mid A=a\}$ follows
from $E\{R_{a,a'}(M,\mathbf X)\mid A=a\}=1$. For a binary mediator,
conditioning on $\mathbf X$ gives
\[
E\{R_{a,a'}^2(M,\mathbf X)\mid \mathbf X,A=a\}
=
\frac{p_{a'}^2(\mathbf X)}{p_a(\mathbf X)}
+
\frac{\{1-p_{a'}(\mathbf X)\}^2}{1-p_a(\mathbf X)},
\]
and averaging over $\mathbf X$ gives the displayed expression for
$d_{a,a'}$ in the main text.

For the direct effect, the influence functions for $\theta(1,a)$ and
$\theta(0,a)$ are supported on disjoint treatment arms. Their covariance is
therefore zero, and substituting the arm-specific variance decomposition gives
\eqref{eq:direct-design-effect-var}. For the indirect effect, both components
are estimated within arm $a$, so the influence-function contrast is
\[
\frac{I(A=a)}{\pi_a}
\left[
R_{a,1}(M,\mathbf X)\{Y-\theta(a,1)\}
-
R_{a,0}(M,\mathbf X)\{Y-\theta(a,0)\}
\right].
\]
Taking its second moment gives \eqref{eq:indirect-general-var}. Expanding the
square and applying the same decomposition to the two diagonal terms gives
\eqref{eq:indirect-design-effect-var}.

To establish the nonnegativity of $K_a$, note that $R_{a,a}(M,\mathbf X)=1$
and, for $a'\in\{0,1\}$,
\begin{equation*}
E\left[
R_{a,a'}(M,\mathbf X)\{Y-\theta(a,a')\}
\mid A=a
\right]=0.
\end{equation*}
Using the symmetry of covariance, it follows that
\begin{equation*}
\begin{aligned}
K_a
&=
E\left[
R_{a,1-a}(M,\mathbf X)
\{Y-\theta(a,1-a)\}
\{Y-\theta(a,a)\}
\mid A=a
\right]\\
&=
E\left[
R_{a,1-a}(M,\mathbf X)
\{Y-\theta(a,1-a)\}^2
\mid A=a
\right]\\
&\quad+
\{\theta(a,1-a)-\theta(a,a)\}
E\left[
R_{a,1-a}(M,\mathbf X)
\{Y-\theta(a,1-a)\}
\mid A=a
\right]\\
&=
E\left[
R_{a,1-a}(M,\mathbf X)
\{Y-\theta(a,1-a)\}^2
\mid A=a
\right]
\geq 0,
\end{aligned}
\end{equation*}
because the mediator probability ratio is nonnegative.

Finally, suppose that for the target contrast $\tau(a)$,
\[
\sqrt N\{\hat\tau(a)-\tau(a)\}
\stackrel{d}{\to}N\{0,V_{\tau_*}(a)\}
\]
under the design alternative $\tau(a)=\tau_*(a)$. For a one-sided test in the
prespecified direction, the Wald statistic is approximately normal with unit
variance and mean $\sqrt N|\tau_*(a)|/\sqrt{V_{\tau_*}(a)}$. Thus the approximate
power is
\[
\Phi\left(
\frac{\sqrt N|\tau_*(a)|}{\sqrt{V_{\tau_*}(a)}}-z_{1-\alpha}
\right).
\]
Equating this expression to $1-\beta$ and rounding upward yields
\eqref{eq:wald-sample-size}.
\end{proof}

\subsection{A Lyapunov central limit theorem for baseline-covariate scores}
\label{app:lyapunov-score}

The following result restates Theorem~1 of \citet{liu2026sample}. It gives a
sufficient Lyapunov condition under which a weighted covariate score is
asymptotically Gaussian.
\begin{lemma}
\label{lem:lyapunov-score}
Let $X_1,X_2,\ldots$ be independent random variables, each with mean 0 and
variance 1. Let $\beta_1,\beta_2,\ldots$ be real coefficients. Suppose that
there exists a positive constant $B<\infty$ such that $E(X_j^4)<B$ for every
$j$, and
\begin{equation*}
\frac{
\max_{1\leq j\leq p}\beta_j^2
}{
\sum_{j=1}^p\beta_j^2
}
=
o(1).
\end{equation*}
Then, as $p\rightarrow\infty$,
\begin{equation*}
\frac{
\sum_{j=1}^p\beta_jX_j
}{
\left(\sum_{j=1}^p\beta_j^2\right)^{1/2}
}
\overset{d}{\longrightarrow}
N(0,1).
\end{equation*}
\end{lemma}

A corresponding multivariate approximation can be obtained by applying the
same Lyapunov argument to every fixed linear combination and invoking the
Cram\'er--Wold device, provided that the limiting covariance matrix exists.
Here these results motivate the Gaussian working approximations; they are not
used as identification assumptions.

\subsection{Derivation of the continuous-outcome design-stage variances}
\label{app:continuous-outcome-working-derivation}

To establish the parameterization in
\eqref{eq:binary-score-parameters}, recall that
$W\sim N(\mu_W,\sigma_W^2)$ and
$\epsilon_M\sim N(0,1)$ are independent. Therefore,
\[
W+a\Delta_M+\epsilon_M
\sim
N(\mu_W+a\Delta_M,1+\sigma_W^2),
\]
and hence
\[
q_a
=
P\{M(a)=1\}
=
P(W+a\Delta_M+\epsilon_M>0)
=
\Phi\left(
\frac{\mu_W+a\Delta_M}{\sqrt{1+\sigma_W^2}}
\right).
\]
Writing $z_a=\Phi^{-1}(q_a)$ gives
\[
\mu_W+a\Delta_M=\sqrt{1+\sigma_W^2}\,z_a.
\]
Moreover,
\[
R_M^2=\frac{\sigma_W^2}{1+\sigma_W^2}
\quad\Longrightarrow\quad
\sigma_W^2=\frac{R_M^2}{1-R_M^2},
\qquad
\sqrt{1+\sigma_W^2}=\frac{1}{\sqrt{1-R_M^2}}.
\]
Setting $a=0$ and then subtracting the equation for $a=0$ from that for
$a=1$ yields the expressions for $\mu_W$ and $\Delta_M$ in
\eqref{eq:binary-score-parameters}.

Next,
\[
\xi_a=\operatorname{Cov}\{\Phi(W+a\Delta_M),W\}.
\]
Stein's identity gives
\[
\xi_a=\sigma_W^2E\{\phi(W+a\Delta_M)\}.
\]
Because $W+a\Delta_M$ is normal with mean $\mu_W+a\Delta_M$ and variance
$\sigma_W^2$,
\[
E\{\phi(W+a\Delta_M)\}
=
\frac{1}{\sqrt{1+\sigma_W^2}}
\phi\left(
\frac{\mu_W+a\Delta_M}{\sqrt{1+\sigma_W^2}}
\right),
\]
Substituting \eqref{eq:binary-score-parameters} yields
\[
\xi_a
=
\frac{R_M^2}{\sqrt{1-R_M^2}}
\phi\{\Phi^{-1}(q_a)\}.
\]

The formulas for $b_a$ and $\sigma_{v,a}^2$ follow by matching the two
design inputs $S_a^2$ and $\kappa_a$. Under
\eqref{eq:continuous-outcome-projection} with $m=M(a)$,
\[
Y(a,M(a))
=
E\{f_a(\mathbf X)\}+\gamma_aq_a
+b_a(W-\mu_W)+\gamma_a\{M(a)-q_a\}+v_a .
\]
Because $E\{v_a\mid W,M(a)\}=0$, the residual is uncorrelated with both
$W$ and $M(a)$. Hence
\[
\operatorname{Cov}\{Y(a,M(a)),W\}
=
b_a\sigma_W^2+\gamma_a\xi_a .
\]
On the other hand,
\[
\operatorname{Cov}\{Y(a,M(a)),W\}
=
\kappa_a S_a\sigma_W .
\]
For $R_M^2>0$, equating the two expressions gives
\[
b_a
=
\frac{\kappa_a S_a\sigma_W-\gamma_a\xi_a}{\sigma_W^2}.
\]
Similarly, using
$\operatorname{Var}\{M(a)\}=q_a(1-q_a)$,
$\operatorname{Cov}\{M(a),W\}=\xi_a$, and
$E\{v_a^2\mid W,M(a)\}=\sigma_{v,a}^2$, we obtain
\[
\begin{split}
S_a^2
&=
\operatorname{Var}\{Y(a,M(a))\}
\\
&=
b_a^2\sigma_W^2
+
\gamma_a^2 q_a(1-q_a)
+
2b_a\gamma_a\xi_a
+
\sigma_{v,a}^2 .
\end{split}
\]
Solving for the residual variance yields
\[
\sigma_{v,a}^2
=
S_a^2
-b_a^2\sigma_W^2
-\gamma_a^2 q_a(1-q_a)
-2b_a\gamma_a\xi_a .
\]
When $R_M^2=0$, $W$ is degenerate, $\kappa_a$ is not required, and the
projection term $b_a(W-\mu_W)$ is omitted. Thus $\xi_a=0$ and
\[
\sigma_{v,a}^2
=
S_a^2-\gamma_a^2q_a(1-q_a).
\]

Under the working projection \eqref{eq:continuous-outcome-projection},
\[
Y(a,m)-\theta(a,a')=h_{a,a'}(m,W)+v_a.
\]
The assumptions on $v_a$ imply the conditional second moment
\[
E\left[
\{Y(a,m)-\theta(a,a')\}^2\mid W,M(a)=m
\right]
=
h_{a,a'}^2(m,W)+\sigma_{v,a}^2.
\]
Because
$h_{a,a'}(m,W)=h_{a,a}(m,W)+\gamma_a(q_a-q_{a'})$,
$E\{h_{a,a}(M,W)\mid A=a\}=0$, and
$S_a^2=E\{h_{a,a}^2(M,W)\mid A=a\}+\sigma_{v,a}^2$, it follows that
\[
B_{a,a'}
=
S_a^2+\gamma_a^2(q_a-q_{a'})^2.
\]
Moreover, since $R_{a,a'}^W(M,W)$ is determined by $(W,M)$,
\[
\begin{aligned}
C_{a,a'}
&=
E_W\left[
\sum_{m=0}^1g_a^W(m\mid W)
\left\{\{R_{a,a'}^W(m,W)\}^2-d_{a,a'}\right\}
\left\{h_{a,a'}^2(m,W)+\sigma_{v,a}^2\right\}
\right]\\
&=
E_W\left[
\sum_{m=0}^1g_a^W(m\mid W)
\left\{\{R_{a,a'}^W(m,W)\}^2-d_{a,a'}\right\}
h_{a,a'}^2(m,W)
\right].
\end{aligned}
\]
Here conditional mean zero removes the cross term involving $v_a$, while
conditional homoscedasticity makes its second-moment contribution constant;
the latter cancels because
$E[\{R_{a,a'}^W(M,W)\}^2-d_{a,a'}\mid A=a]=0$. This gives
\eqref{eq:continuous-C}.

The same conditional-moment calculation gives
\[
K_a
=
E_W\left[
\sum_{m=0}^1g_a^W(m\mid W)R_{a,1-a}^W(m,W)
\left\{
h_{a,1-a}(m,W)h_{a,a}(m,W)+\sigma_{v,a}^2
\right\}
\right].
\]
Since $g_a^WR_{a,1-a}^W=g_{1-a}^W$ and
$h_{a,a}=h_{a,1-a}+\gamma_a(q_{1-a}-q_a)$, the last display also equals
\[
E_W\left[
\sum_{m=0}^1g_{1-a}^W(m\mid W)
\{h_{a,1-a}^2(m,W)+\sigma_{v,a}^2\}
\right]
\geq0,
\]
because the term linear in $h_{a,1-a}$ has mean zero under
$g_{1-a}^W$.

Substituting the conditional second moment above into
\eqref{eq:rmpw-theta-var} gives
the mediation-mean design-stage variance. Finally,
$R_{a,a}=1$ implies $d_{a,a}=1$, $C_{a,a}=0$, and
$B_{a,a}=S_a^2$. Substitution into
\eqref{eq:direct-design-effect-var} and
\eqref{eq:indirect-design-effect-var}, together with the displayed expression
for $K_a$, gives the NDE and NIE design-stage variances.

\subsection{Derivation of the binary-outcome design-stage variances}
\label{app:proof-binary-outcome-working}

\textit{Latent-scale reduction and calibration.}
For $R_M^2>0$, the population projection residual is
\[
v_a^*
=
f_a^*(\mathbf X)-E\{f_a^*(\mathbf X)\}
-\widetilde b_a(W-\mu_W).
\]
Joint normality of $f_a^*(\mathbf X)$ and $W$ implies that $v_a^*$ is
independent of $W$. Because the latent mediator error is independent of
$\mathbf X$, the mediator model further implies that $v_a^*$ is independent
of $M(a)$ conditional on $W$. The latent outcome error satisfies
\[
\epsilon_a^*\sim\mathcal N(0,1),
\qquad
\epsilon_a^*\indep\{\mathbf X,M(a)\}.
\]
Because $v_a^*$ is Gaussian and independent of $M(a)$ conditional on $W$,
it follows that
\[
u_a^*\mid\{W,M(a)\}
\sim
\mathcal N(0,\sigma_{u,a}^2),
\qquad
\sigma_{u,a}^2
=
\operatorname{Var}(v_a^*)+1.
\]
Thresholding the latent
outcome gives
\[
\begin{split}
P(Y=1\mid A=a,M=m,W=w)
&=
\Phi\left\{
\frac{E\{f_a^*(\mathbf X)\}}{\sigma_{u,a}}
+
\frac{\widetilde b_a}{\sigma_{u,a}}(w-\mu_W)
+
\frac{\gamma_a^*}{\sigma_{u,a}}(m-q_a)
\right\}
\\
&=
\rho_a(w,m).
\end{split}
\]
When $R_M^2=0$, $W$ is degenerate and the projection term is omitted.

For $R_M^2>0$, the three standardized quantities
$E\{f_a^*(\mathbf X)\}/\sigma_{u,a}$, $b_a^\rho$, and $\gamma_a^\rho$
are jointly calibrated by solving
\[
\begin{aligned}
E_W\left[
\sum_{m=0}^1g_a^W(m\mid W)\rho_a(W,m)
\right]
&=
\theta(a,a),
\\
\frac{
E_W\left[
(W-\mu_W)\sum_{m=0}^1g_a^W(m\mid W)\rho_a(W,m)
\right]
}{
\sigma_W[\theta(a,a)\{1-\theta(a,a)\}]^{1/2}
}
&=
\kappa_a,
\\
E_W\left[
\{p_1^W(W)-p_0^W(W)\}
\{\rho_a(W,1)-\rho_a(W,0)\}
\right]
&=
\delta(a).
\end{aligned}
\]
These one-dimensional expectations can be evaluated by numerical Gaussian
quadrature, and the equations can then be solved by standard nonlinear
root-finding whenever the specified inputs admit a compatible solution.
When $R_M^2=0$, the second equation and the term
$b_a^\rho(W-\mu_W)$ are omitted. If $q_1=q_0$, the mediator distributions
coincide and the third equation requires $\delta(a)=0$.

For the two-arm NDE calculation used in the numerical study, the working-
variance inputs are the two observed-arm event probabilities, outcome--$W$
correlations, and arm-specific NIEs. The two probit laws are calibrated
separately using $\theta(a,a)$, $\kappa_a$, and $\delta(a)$ in the system
above. The cross-arm mediation means and NDEs then follow from
\[
\begin{aligned}
\theta(0,1)&=\theta(0,0)+\delta(0),
&\qquad
\theta(1,0)&=\theta(1,1)-\delta(1),
\\
\zeta(0)&=\theta(1,0)-\theta(0,0),
&
\zeta(1)&=\theta(1,1)-\theta(0,1).
\end{aligned}
\]
Thus, the NDE is derived after calibration rather than imposed as an
additional calibration target. All implied mediation probabilities must lie
in $[0,1]$, and both numerical systems must admit compatible solutions.

\textit{Componentwise calculations.}
Conditional on $(W=w,M=m,A=a)$, $Y$ is Bernoulli with event probability
$\rho_a(w,m)$. Hence
\[
\begin{split}
&E\left[
\{Y-\theta(a,a')\}^2
\mid W=w,M=m,A=a
\right]
\\
&\qquad=
\{\rho_a(w,m)-\theta(a,a')\}^2
+
\rho_a(w,m)\{1-\rho_a(w,m)\}
\\
&\qquad=
\theta^2(a,a')
+
\{1-2\theta(a,a')\}\rho_a(w,m).
\end{split}
\]
Averaging over the mediator distribution in arm $a$ gives
\[
\begin{split}
B_{a,a'}
&=
\theta^2(a,a')
+
\{1-2\theta(a,a')\}\theta(a,a)
\\
&=
\theta(a,a)\{1-\theta(a,a)\}
+
\{\theta(a,a)-\theta(a,a')\}^2,
\end{split}
\]
which proves \eqref{eq:binary-B}.

Because
$E[\{R_{a,a'}^W(M,W)\}^2-d_{a,a'}\mid A=a]=0$, the covariance component is
\[
\begin{split}
C_{a,a'}
&=
E_W\Bigg[
\sum_{m=0}^1
g_a^W(m\mid W)
\left\{
\{R_{a,a'}^W(m,W)\}^2-d_{a,a'}
\right\}
\\
&\hspace{30mm}\times
E\left[
\{Y-\theta(a,a')\}^2
\mid W,M=m,A=a
\right]
\Bigg]
\\
&=
\{1-2\theta(a,a')\}
E_W\Bigg[
\sum_{m=0}^1
g_a^W(m\mid W)
\left\{
\{R_{a,a'}^W(m,W)\}^2-d_{a,a'}
\right\}
\rho_a(W,m)
\Bigg],
\end{split}
\]
which proves \eqref{eq:binary-C}; the constant term $\theta^2(a,a')$
vanishes because the centered squared ratio has mean zero.

For the shared-arm component, the residual-centering identity established in
Section~\ref{sec:variance-decomposition} gives
\[
\begin{split}
K_a
&=
E\left[
R_{a,1-a}^W(M,W)
\{Y-\theta(a,1-a)\}
\{Y-\theta(a,a)\}
\mid A=a
\right]
\\
&=
E\left[
R_{a,1-a}^W(M,W)
\{Y-\theta(a,1-a)\}^2
\mid A=a
\right].
\end{split}
\]
Since
$g_a^W(m\mid W)R_{a,1-a}^W(m,W)=g_{1-a}^W(m\mid W)$,
\[
\begin{split}
K_a
=
E_W\Bigg[
\sum_{m=0}^1
g_{1-a}^W(m\mid W)
\Big[&
\{\rho_a(W,m)-\theta(a,1-a)\}^2
+
\rho_a(W,m)\{1-\rho_a(W,m)\}
\Big]
\Bigg].
\end{split}
\]
Under this reweighted law, $Y$ is binary with mean $\theta(a,1-a)$.
Therefore,
\[
K_a
=
\theta(a,1-a)\{1-\theta(a,1-a)\},
\]
which proves \eqref{eq:binary-K}.

Finally, Proposition~\ref{prop:oracle-rmpw-var} and
\eqref{eq:theta-design-effect-decomp} give
\[
V_\theta^B(a,a')
=
\frac{d_{a,a'}B_{a,a'}+C_{a,a'}}{\pi_a},
\]
which is the mediation-mean design-stage variance. When $a'=a$,
$R_{a,a}^W=1$, so $d_{a,a}=1$, $C_{a,a}=0$, and
$B_{a,a}=\theta(a,a)\{1-\theta(a,a)\}$. Substitution into
\eqref{eq:direct-design-effect-var} and
\eqref{eq:indirect-design-effect-var}, together with
\eqref{eq:binary-K}, gives the NDE and NIE design-stage variances.

\section{Simulation settings and additional oracle-ratio results}
\label{app:additional-oracle-results}

This section provides the complete simulation-panel definitions and the
oracle-ratio results under the mixed-covariate designs. The corresponding
Gaussian-covariate results are reported in the main text. For the mixed-covariate design, the mutually independent components of
$\mathbf X$ are generated as follows before standardization:
$X_j\sim\operatorname{Bernoulli}(p_j)$ for $j=1,\ldots,4$, where
$(p_1,p_2,p_3,p_4)=(0.2,0.4,0.6,0.8)$;
$X_5\sim\operatorname{Uniform}(0,1)$;
$X_j\sim\operatorname{Poisson}(j-5)$ for $j=6,7,8$;
$X_9\sim\operatorname{Gamma}(\text{shape}=2,\text{scale}=3)$; and
$X_{10}\sim\operatorname{Beta}(2,3)$. Each component is then centered and
scaled to have mean zero and variance one, with the notation $X_j$ retained
for the standardized variable.

For both covariate designs, define the fixed orthonormal vectors
\[
u_M=\frac{1}{\sqrt{10}}(1,\ldots,1)^\top,
\qquad
u_\perp=\frac{1}{\sqrt{82.5}}
(-4.5,-3.5,-2.5,-1.5,-0.5,0.5,1.5,2.5,3.5,4.5)^\top.
\]
For given values of $R_M^2$ and $\rho$, the coefficient vectors are
constructed as
\[
\beta_M
=
\left(\frac{R_M^2}{1-R_M^2}\right)^{1/2}u_M,
\qquad
\beta_Y
=
0.4\left\{\rho u_M+(1-\rho^2)^{1/2}u_\perp\right\}.
\]
Consequently,
$
\operatorname{Var}(\mathbf X^\top\beta_M)
=
\frac{R_M^2}{1-R_M^2},
$
and, when $R_M^2>0$,
$
\operatorname{Corr}
(\mathbf X^\top\beta_M,\mathbf X^\top\beta_Y)=\rho.
$
When $R_M^2=0$, we set $\beta_M=0$, in which case the mediator score is
constant and the correlation is not defined. The same coefficient
construction is used for the continuous- and binary-outcome simulations.
\subsection{Continuous outcome}
\begin{table}[htbp]
\centering
\footnotesize
\caption{Continuous-outcome simulation panels.}
\label{tab:continuous-panels}
\setlength{\tabcolsep}{4pt}
\resizebox{\textwidth}{!}{%
\begin{tabular}{lll}
\toprule
Panel & Design feature varied & Values \\
\midrule
Baseline & Baseline setting &
$\eta_1-\eta_0=0.08,\ \gamma_0=\gamma_1=0.4,\ q_0=0.35,\ q_1=0.55,\ R_M^2=0.25,\ \rho=0.5$ \\
A & Indirect pathway strength &
$\gamma_0=\gamma_1\in\{0.2,0.4,0.6\}$ \\
B & Direct pathway strength &
$\eta_1-\eta_0\in\{0.04,0.08,0.12\}$ \\
C & Mediator prevalence &
$\{q_0,q_1\}\in\{\{0.40,0.60\},\{0.20,0.40\},\{0.60,0.80\}\}$ \\
D & Mediator heterogeneity &
$R_M^2\in\{0,0.25,0.50\}$ \\
E & Prognostic-score correlation &
$\rho\in\{0,0.5,0.8\}$ \\
F & Treatment--mediator interaction &
$\{\gamma_0,\gamma_1\}\in\{\{0.3,0.5\},\{0.2,0.6\}\}$ \\
\bottomrule
\end{tabular}
}
\end{table}

\begin{table}[htbp]
\centering
\scriptsize
\caption{Oracle-ratio continuous-outcome simulations with mixed covariates.}
\label{tab:continuous-mixed}
\setlength{\tabcolsep}{2pt}
\resizebox{\textwidth}{!}{%
\begin{tabular}{l*{12}{c}}
\simtableheader
Baseline & \simresult{438}{459}{0.812} & \simresult{477}{474}{0.810} & \simresult{5250}{5109}{0.818} & \simresult{5170}{5058}{0.788} \\
$\gamma_0=\gamma_1=0.2$ & \simresult{1755}{1816}{0.808} & \simresult{1915}{1934}{0.814} & \simresult{5068}{5165}{0.766} & \simresult{5014}{5066}{0.790} \\
$\gamma_0=\gamma_1=0.4$ & \simresult{436}{435}{0.774} & \simresult{476}{480}{0.780} & \simresult{5244}{5169}{0.834} & \simresult{5164}{5006}{0.818} \\
$\gamma_0=\gamma_1=0.6$ & \simresult{194}{188}{0.786} & \simresult{211}{218}{0.782} & \simresult{5506}{5765}{0.826} & \simresult{5387}{5501}{0.810} \\
$\eta_1-\eta_0=0.04$ & \simresult{436}{439}{0.806} & \simresult{476}{460}{0.798} & \simresult{20950}{23089}{0.796} & \simresult{20631}{20511}{0.786} \\
$\eta_1-\eta_0=0.08$ & \simresult{437}{440}{0.792} & \simresult{477}{465}{0.826} & \simresult{5245}{5241}{0.786} & \simresult{5165}{5142}{0.812} \\
$\eta_1-\eta_0=0.12$ & \simresult{437}{454}{0.820} & \simresult{476}{434}{0.798} & \simresult{2331}{2310}{0.814} & \simresult{2295}{2319}{0.800} \\
$\{q_0,q_1\}=\{0.40,0.60\}$ & \simresult{449}{436}{0.822} & \simresult{450}{450}{0.808} & \simresult{5204}{5182}{0.750} & \simresult{5204}{4979}{0.792} \\
$\{q_0,q_1\}=\{0.20,0.40\}$ & \simresult{462}{488}{0.788} & \simresult{702}{680}{0.798} & \simresult{5493}{5777}{0.802} & \simresult{5089}{5059}{0.824} \\
$\{q_0,q_1\}=\{0.60,0.80\}$ & \simresult{702}{726}{0.798} & \simresult{462}{481}{0.840} & \simresult{5083}{4739}{0.786} & \simresult{5487}{5404}{0.784} \\
$R_M^2=0$ & \simresult{363}{374}{0.780} & \simresult{395}{398}{0.798} & \simresult{5038}{4982}{0.810} & \simresult{4982}{4818}{0.822} \\
$R_M^2=0.25$ & \simresult{436}{431}{0.816} & \simresult{476}{477}{0.800} & \simresult{5244}{5010}{0.842} & \simresult{5164}{4902}{0.830} \\
$R_M^2=0.50$ & \simresult{586}{568}{0.816} & \simresult{643}{670}{0.814} & \simresult{5462}{5280}{0.796} & \simresult{5359}{5414}{0.828} \\
$\rho=0$ & \simresult{448}{425}{0.820} & \simresult{490}{497}{0.798} & \simresult{5129}{5110}{0.812} & \simresult{5062}{5103}{0.828} \\
$\rho=0.5$ & \simresult{437}{436}{0.794} & \simresult{477}{491}{0.768} & \simresult{5248}{5232}{0.794} & \simresult{5168}{5148}{0.812} \\
$\rho=0.8$ & \simresult{429}{450}{0.782} & \simresult{466}{465}{0.810} & \simresult{5311}{5418}{0.798} & \simresult{5224}{5185}{0.818} \\
$\{\gamma_0,\gamma_1\}=\{0.30,0.50\}$ & \simresult{279}{256}{0.800} & \simresult{847}{872}{0.810} & \simresult{930}{919}{0.820} & \simresult{1468}{1529}{0.818} \\
$\{\gamma_0,\gamma_1\}=\{0.20,0.60\}$ & \simresult{194}{185}{0.820} & \simresult{1916}{1985}{0.796} & \simresult{375}{360}{0.808} & \simresult{683}{681}{0.802} \\
\bottomrule
\end{tabular}
}
\end{table}

\subsection{Binary outcome}

\begin{table}[htbp]
\centering
\footnotesize
\caption{Binary-outcome simulation panels.}
\label{tab:binary-panels}
\setlength{\tabcolsep}{4pt}
\resizebox{\textwidth}{!}{%
\begin{tabular}{lll}
\toprule
Panel & Design feature varied & Values \\
\midrule
Baseline & Baseline setting &
$\eta_1-\eta_0=0.25,\ \gamma_0=\gamma_1=0.35,\ q_0=0.35,\ q_1=0.55,\ R_M^2=0.25,\ \rho=0.5,\ r_0=0.30$ \\
A & Indirect pathway strength &
$\gamma_0=\gamma_1\in\{0.20,0.35,0.50\}$ \\
B & Direct pathway strength &
$\eta_1-\eta_0\in\{0.15,0.25,0.35\}$ \\
C & Mediator prevalence &
$\{q_0,q_1\}\in\{\{0.40,0.60\},\{0.20,0.40\},\{0.60,0.80\}\}$ \\
D & Mediator heterogeneity &
$R_M^2\in\{0,0.25,0.50\}$ \\
E & Prognostic-score correlation &
$\rho\in\{0,0.5,0.8\}$ \\
F & Treatment--mediator interaction &
$\{\gamma_0,\gamma_1\}\in\{\{0.25,0.45\},\{0.15,0.55\}\}$ \\
G & Outcome prevalence &
$r_0\in\{0.15,0.30,0.50\}$ \\
\bottomrule
\end{tabular}
}
\end{table}

\begin{table}[htbp]
\centering
\scriptsize
\caption{Oracle-ratio binary-outcome simulations with mixed covariates.}
\label{tab:binary-mixed}
\setlength{\tabcolsep}{2pt}
\resizebox{\textwidth}{!}{%
\begin{tabular}{l*{12}{c}}
\simtableheader
Baseline & \simresult{928}{881}{0.748} & \simresult{1094}{1093}{0.814} & \simresult{889}{865}{0.848} & \simresult{883}{861}{0.804} \\
$\gamma_0=\gamma_1=0.20$ & \simresult{2865}{2711}{0.814} & \simresult{3323}{3472}{0.806} & \simresult{864}{859}{0.790} & \simresult{863}{901}{0.788} \\
$\gamma_0=\gamma_1=0.35$ & \simresult{925}{952}{0.786} & \simresult{1091}{1059}{0.818} & \simresult{888}{912}{0.772} & \simresult{882}{898}{0.780} \\
$\gamma_0=\gamma_1=0.50$ & \simresult{449}{481}{0.772} & \simresult{539}{540}{0.792} & \simresult{922}{914}{0.814} & \simresult{912}{880}{0.818} \\
$\eta_1-\eta_0=0.15$ & \simresult{944}{866}{0.804} & \simresult{1091}{1081}{0.780} & \simresult{2504}{2494}{0.810} & \simresult{2477}{2477}{0.800} \\
$\eta_1-\eta_0=0.25$ & \simresult{927}{880}{0.792} & \simresult{1093}{1088}{0.838} & \simresult{888}{888}{0.804} & \simresult{882}{888}{0.792} \\
$\eta_1-\eta_0=0.35$ & \simresult{914}{947}{0.780} & \simresult{1093}{1065}{0.810} & \simresult{446}{432}{0.808} & \simresult{445}{426}{0.824} \\
$\{q_0,q_1\}=\{0.40,0.60\}$ & \simresult{955}{1000}{0.784} & \simresult{1036}{1038}{0.786} & \simresult{881}{866}{0.808} & \simresult{886}{872}{0.818} \\
$\{q_0,q_1\}=\{0.20,0.40\}$ & \simresult{961}{943}{0.818} & \simresult{1566}{1539}{0.814} & \simresult{932}{927}{0.802} & \simresult{877}{911}{0.790} \\
$\{q_0,q_1\}=\{0.60,0.80\}$ & \simresult{1473}{1576}{0.772} & \simresult{1055}{1005}{0.806} & \simresult{857}{863}{0.790} & \simresult{921}{900}{0.820} \\
$R_M^2=0$ & \simresult{754}{751}{0.804} & \simresult{907}{944}{0.778} & \simresult{852}{835}{0.774} & \simresult{845}{848}{0.794} \\
$R_M^2=0.25$ & \simresult{926}{918}{0.812} & \simresult{1092}{1019}{0.826} & \simresult{888}{871}{0.788} & \simresult{882}{871}{0.776} \\
$R_M^2=0.50$ & \simresult{1242}{1241}{0.800} & \simresult{1471}{1487}{0.778} & \simresult{927}{971}{0.776} & \simresult{922}{912}{0.816} \\
$\rho=0$ & \simresult{923}{901}{0.784} & \simresult{1122}{1081}{0.794} & \simresult{865}{847}{0.772} & \simresult{857}{841}{0.776} \\
$\rho=0.5$ & \simresult{927}{880}{0.778} & \simresult{1093}{1082}{0.822} & \simresult{889}{886}{0.842} & \simresult{883}{885}{0.764} \\
$\rho=0.8$ & \simresult{922}{893}{0.798} & \simresult{1067}{1028}{0.814} & \simresult{901}{962}{0.798} & \simresult{897}{863}{0.816} \\
$\{\gamma_0,\gamma_1\}=\{0.25,0.45\}$ & \simresult{550}{532}{0.814} & \simresult{2128}{2149}{0.806} & \simresult{404}{410}{0.796} & \simresult{512}{523}{0.786} \\
$\{\gamma_0,\gamma_1\}=\{0.15,0.55\}$ & \simresult{365}{384}{0.766} & \simresult{5892}{5741}{0.810} & \simresult{231}{240}{0.792} & \simresult{337}{343}{0.798} \\
$r_0=0.15$ & \simresult{1136}{1141}{0.794} & \simresult{1488}{1430}{0.812} & \simresult{1134}{1135}{0.784} & \simresult{1138}{1136}{0.824} \\
$r_0=0.30$ & \simresult{923}{1001}{0.770} & \simresult{1089}{1066}{0.824} & \simresult{888}{885}{0.828} & \simresult{882}{865}{0.778} \\
$r_0=0.50$ & \simresult{932}{965}{0.770} & \simresult{986}{921}{0.766} & \simresult{845}{764}{0.786} & \simresult{833}{766}{0.794} \\
\bottomrule
\end{tabular}
}
\end{table}

\section{Estimated-ratio simulation results}
\label{app:estimated-ratio-results}

This section reports the estimated-ratio results under both covariate designs.
In these analyses, the mediator probit model is refit within each
simulated trial before constructing the mediator probability ratios. The
analytic sample size $N$ is the same design-stage sample size as in the main
oracle-ratio analysis, while $N_{\mathrm{sim}}$ is the
simulation-based benchmark obtained under the estimated-ratio implementation.
For each estimand, $N_{\mathrm{sim}}$ is computed from an
estimand-specific local adaptive grid that brackets 0.80, followed by
monotone interpolation of the empirical power curve.
The reported power at $N$ is obtained from 500 trials at the design-stage
sample size. Wald standard errors use the same influence-function plug-in
calculation as in the oracle analysis, with the fitted mediator probability
ratios treated as fixed.

\subsection{Continuous outcome}

\begin{table}[!htbp]
\centering
\footnotesize
\caption{Robustness checks for continuous-outcome simulations.}
\label{tab:continuous-robustness}
\setlength{\tabcolsep}{4pt}
\resizebox{\textwidth}{!}{%
\begin{tabular}{lcccc}
\toprule
Analysis & Mean power (SD) & Power range & Indirect-effect null (SD) & Direct-effect null (SD) \\
\midrule
Oracle mediator ratio & 0.802 (0.017) & 0.750--0.846 & 0.056 (0.010) & 0.050 (0.011) \\
Estimated mediator ratio & 0.799 (0.019) & 0.748--0.846 & 0.051 (0.009) & 0.046 (0.010) \\
\bottomrule
\end{tabular}
}
\end{table}

\begin{table}[!htbp]
\centering
\scriptsize
\caption{Estimated-ratio continuous-outcome simulations with Gaussian covariates.}
\label{tab:app-est-continuous-gaussian}
\setlength{\tabcolsep}{2pt}
\resizebox{\textwidth}{!}{%
\begin{tabular}{l*{12}{c}}
\simtableheader
Baseline & \simresult{437}{435}{0.768} & \simresult{477}{462}{0.800} & \simresult{5246}{5309}{0.814} & \simresult{5166}{5413}{0.838} \\
$\gamma_0=\gamma_1=0.2$ & \simresult{1757}{1734}{0.790} & \simresult{1916}{1886}{0.812} & \simresult{5066}{4925}{0.822} & \simresult{5013}{4776}{0.794} \\
$\gamma_0=\gamma_1=0.4$ & \simresult{436}{422}{0.790} & \simresult{476}{459}{0.788} & \simresult{5235}{5185}{0.826} & \simresult{5155}{5155}{0.812} \\
$\gamma_0=\gamma_1=0.6$ & \simresult{194}{193}{0.784} & \simresult{212}{206}{0.802} & \simresult{5514}{5589}{0.806} & \simresult{5395}{5611}{0.806} \\
$\eta_1-\eta_0=0.04$ & \simresult{436}{434}{0.784} & \simresult{475}{471}{0.794} & \simresult{20963}{20160}{0.804} & \simresult{20644}{21383}{0.804} \\
$\eta_1-\eta_0=0.08$ & \simresult{438}{455}{0.826} & \simresult{477}{483}{0.778} & \simresult{5246}{5435}{0.820} & \simresult{5166}{5032}{0.798} \\
$\eta_1-\eta_0=0.12$ & \simresult{437}{436}{0.788} & \simresult{476}{468}{0.804} & \simresult{2331}{2368}{0.840} & \simresult{2296}{2286}{0.798} \\
$\{q_0,q_1\}=\{0.40,0.60\}$ & \simresult{449}{472}{0.796} & \simresult{449}{442}{0.798} & \simresult{5195}{5129}{0.816} & \simresult{5195}{5108}{0.790} \\
$\{q_0,q_1\}=\{0.20,0.40\}$ & \simresult{462}{454}{0.788} & \simresult{701}{713}{0.812} & \simresult{5497}{5381}{0.804} & \simresult{5092}{5015}{0.772} \\
$\{q_0,q_1\}=\{0.60,0.80\}$ & \simresult{702}{712}{0.806} & \simresult{462}{491}{0.790} & \simresult{5086}{5595}{0.808} & \simresult{5491}{5148}{0.776} \\
$R_M^2=0$ & \simresult{363}{369}{0.794} & \simresult{394}{398}{0.778} & \simresult{5028}{5123}{0.760} & \simresult{4971}{5028}{0.786} \\
$R_M^2=0.25$ & \simresult{436}{442}{0.782} & \simresult{475}{451}{0.796} & \simresult{5240}{5409}{0.804} & \simresult{5160}{4600}{0.798} \\
$R_M^2=0.50$ & \simresult{586}{603}{0.794} & \simresult{642}{627}{0.784} & \simresult{5461}{5735}{0.778} & \simresult{5358}{5179}{0.790} \\
$\rho=0$ & \simresult{449}{426}{0.838} & \simresult{491}{479}{0.790} & \simresult{5129}{5322}{0.824} & \simresult{5062}{5314}{0.800} \\
$\rho=0.5$ & \simresult{436}{458}{0.804} & \simresult{475}{484}{0.808} & \simresult{5236}{5379}{0.794} & \simresult{5156}{5385}{0.820} \\
$\rho=0.8$ & \simresult{428}{441}{0.796} & \simresult{465}{475}{0.762} & \simresult{5309}{5320}{0.800} & \simresult{5222}{5081}{0.810} \\
$\{\gamma_0,\gamma_1\}=\{0.30,0.50\}$ & \simresult{279}{288}{0.794} & \simresult{848}{860}{0.824} & \simresult{933}{882}{0.836} & \simresult{1474}{1469}{0.834} \\
$\{\gamma_0,\gamma_1\}=\{0.20,0.60\}$ & \simresult{194}{208}{0.820} & \simresult{1915}{1994}{0.786} & \simresult{375}{367}{0.784} & \simresult{685}{671}{0.788} \\
\bottomrule
\end{tabular}
}
\end{table}

\begin{table}[!htbp]
\centering
\scriptsize
\caption{Estimated-ratio continuous-outcome simulations with mixed covariates.}
\label{tab:app-est-continuous-mixed}
\setlength{\tabcolsep}{2pt}
\resizebox{\textwidth}{!}{%
\begin{tabular}{l*{12}{c}}
\simtableheader
Baseline & \simresult{438}{445}{0.818} & \simresult{477}{474}{0.810} & \simresult{5250}{5072}{0.792} & \simresult{5170}{5082}{0.824} \\
$\gamma_0=\gamma_1=0.2$ & \simresult{1755}{1768}{0.804} & \simresult{1915}{1944}{0.798} & \simresult{5068}{4928}{0.796} & \simresult{5014}{4868}{0.808} \\
$\gamma_0=\gamma_1=0.4$ & \simresult{436}{449}{0.846} & \simresult{476}{511}{0.776} & \simresult{5244}{5399}{0.776} & \simresult{5164}{5209}{0.802} \\
$\gamma_0=\gamma_1=0.6$ & \simresult{194}{216}{0.790} & \simresult{211}{206}{0.770} & \simresult{5506}{5226}{0.790} & \simresult{5387}{5193}{0.818} \\
$\eta_1-\eta_0=0.04$ & \simresult{436}{431}{0.774} & \simresult{476}{521}{0.778} & \simresult{20950}{20348}{0.820} & \simresult{20631}{20298}{0.806} \\
$\eta_1-\eta_0=0.08$ & \simresult{437}{484}{0.758} & \simresult{477}{494}{0.786} & \simresult{5245}{5384}{0.816} & \simresult{5165}{5170}{0.760} \\
$\eta_1-\eta_0=0.12$ & \simresult{437}{428}{0.800} & \simresult{476}{466}{0.790} & \simresult{2331}{2346}{0.806} & \simresult{2295}{2363}{0.828} \\
$\{q_0,q_1\}=\{0.40,0.60\}$ & \simresult{449}{467}{0.824} & \simresult{450}{437}{0.812} & \simresult{5204}{5013}{0.802} & \simresult{5204}{5001}{0.814} \\
$\{q_0,q_1\}=\{0.20,0.40\}$ & \simresult{462}{460}{0.772} & \simresult{702}{684}{0.766} & \simresult{5493}{5303}{0.806} & \simresult{5089}{5046}{0.812} \\
$\{q_0,q_1\}=\{0.60,0.80\}$ & \simresult{702}{693}{0.790} & \simresult{462}{469}{0.794} & \simresult{5083}{5003}{0.794} & \simresult{5487}{5278}{0.794} \\
$R_M^2=0$ & \simresult{363}{367}{0.788} & \simresult{395}{400}{0.784} & \simresult{5038}{4984}{0.812} & \simresult{4982}{4916}{0.748} \\
$R_M^2=0.25$ & \simresult{436}{468}{0.818} & \simresult{476}{475}{0.804} & \simresult{5244}{5108}{0.812} & \simresult{5164}{4953}{0.800} \\
$R_M^2=0.50$ & \simresult{586}{595}{0.768} & \simresult{643}{673}{0.790} & \simresult{5462}{5352}{0.788} & \simresult{5359}{5418}{0.818} \\
$\rho=0$ & \simresult{448}{439}{0.800} & \simresult{490}{475}{0.780} & \simresult{5129}{5292}{0.788} & \simresult{5062}{4864}{0.786} \\
$\rho=0.5$ & \simresult{437}{444}{0.818} & \simresult{477}{509}{0.808} & \simresult{5248}{5093}{0.800} & \simresult{5168}{4914}{0.820} \\
$\rho=0.8$ & \simresult{429}{434}{0.778} & \simresult{466}{480}{0.784} & \simresult{5311}{5596}{0.762} & \simresult{5224}{5263}{0.810} \\
$\{\gamma_0,\gamma_1\}=\{0.30,0.50\}$ & \simresult{279}{292}{0.784} & \simresult{847}{861}{0.840} & \simresult{930}{950}{0.834} & \simresult{1468}{1425}{0.796} \\
$\{\gamma_0,\gamma_1\}=\{0.20,0.60\}$ & \simresult{194}{201}{0.796} & \simresult{1916}{1864}{0.810} & \simresult{375}{369}{0.828} & \simresult{683}{658}{0.792} \\
\bottomrule
\end{tabular}
}
\end{table}

\subsection{Binary outcome}

\begin{table}[!htbp]
\centering
\footnotesize
\caption{Robustness checks for binary-outcome simulations.}
\label{tab:binary-robustness}
\setlength{\tabcolsep}{4pt}
\resizebox{\textwidth}{!}{%
\begin{tabular}{lcccc}
\toprule
Analysis & Mean power (SD) & Power range & Indirect-effect null (SD) & Direct-effect null (SD) \\
\midrule
Oracle mediator ratio & 0.801 (0.019) & 0.748--0.848 & 0.051 (0.010) & 0.052 (0.011) \\
Estimated mediator ratio & 0.800 (0.018) & 0.758--0.850 & 0.046 (0.009) & 0.050 (0.010) \\
\bottomrule
\end{tabular}
}
\end{table}

\begin{table}[!htbp]
\centering
\scriptsize
\caption{Estimated-ratio binary-outcome simulations with Gaussian covariates.}
\label{tab:app-est-binary-gaussian}
\setlength{\tabcolsep}{2pt}
\resizebox{\textwidth}{!}{%
\begin{tabular}{l*{12}{c}}
\simtableheader
Baseline & \simresult{923}{923}{0.794} & \simresult{1087}{1068}{0.818} & \simresult{885}{896}{0.806} & \simresult{879}{901}{0.826} \\
$\gamma_0=\gamma_1=0.20$ & \simresult{2851}{2800}{0.820} & \simresult{3302}{3147}{0.798} & \simresult{861}{840}{0.804} & \simresult{859}{854}{0.824} \\
$\gamma_0=\gamma_1=0.35$ & \simresult{920}{1022}{0.758} & \simresult{1085}{1104}{0.820} & \simresult{883}{883}{0.784} & \simresult{877}{890}{0.802} \\
$\gamma_0=\gamma_1=0.50$ & \simresult{449}{504}{0.774} & \simresult{538}{527}{0.822} & \simresult{920}{938}{0.828} & \simresult{909}{878}{0.770} \\
$\eta_1-\eta_0=0.15$ & \simresult{938}{1017}{0.804} & \simresult{1084}{1085}{0.822} & \simresult{2495}{2496}{0.798} & \simresult{2466}{2397}{0.796} \\
$\eta_1-\eta_0=0.25$ & \simresult{924}{912}{0.796} & \simresult{1088}{1065}{0.810} & \simresult{885}{872}{0.794} & \simresult{879}{862}{0.764} \\
$\eta_1-\eta_0=0.35$ & \simresult{910}{975}{0.790} & \simresult{1086}{1138}{0.812} & \simresult{445}{440}{0.772} & \simresult{443}{439}{0.814} \\
$\{q_0,q_1\}=\{0.40,0.60\}$ & \simresult{952}{967}{0.808} & \simresult{1031}{985}{0.802} & \simresult{876}{784}{0.788} & \simresult{881}{910}{0.772} \\
$\{q_0,q_1\}=\{0.20,0.40\}$ & \simresult{952}{987}{0.788} & \simresult{1549}{1601}{0.808} & \simresult{930}{986}{0.798} & \simresult{874}{924}{0.792} \\
$\{q_0,q_1\}=\{0.60,0.80\}$ & \simresult{1478}{1632}{0.798} & \simresult{1056}{1098}{0.832} & \simresult{855}{873}{0.788} & \simresult{917}{927}{0.782} \\
$R_M^2=0$ & \simresult{750}{788}{0.780} & \simresult{901}{909}{0.794} & \simresult{848}{909}{0.804} & \simresult{840}{817}{0.784} \\
$R_M^2=0.25$ & \simresult{920}{935}{0.806} & \simresult{1084}{1030}{0.774} & \simresult{884}{905}{0.806} & \simresult{878}{862}{0.792} \\
$R_M^2=0.50$ & \simresult{1234}{1234}{0.768} & \simresult{1461}{1401}{0.802} & \simresult{923}{916}{0.788} & \simresult{917}{871}{0.816} \\
$\rho=0$ & \simresult{923}{952}{0.798} & \simresult{1124}{1102}{0.802} & \simresult{865}{827}{0.788} & \simresult{856}{813}{0.810} \\
$\rho=0.5$ & \simresult{920}{915}{0.804} & \simresult{1085}{1101}{0.806} & \simresult{884}{911}{0.796} & \simresult{877}{912}{0.786} \\
$\rho=0.8$ & \simresult{917}{941}{0.776} & \simresult{1059}{1052}{0.816} & \simresult{897}{893}{0.822} & \simresult{892}{857}{0.796} \\
$\{\gamma_0,\gamma_1\}=\{0.25,0.45\}$ & \simresult{549}{561}{0.798} & \simresult{2118}{2050}{0.822} & \simresult{403}{448}{0.772} & \simresult{510}{499}{0.810} \\
$\{\gamma_0,\gamma_1\}=\{0.15,0.55\}$ & \simresult{364}{400}{0.780} & \simresult{5854}{6045}{0.812} & \simresult{229}{245}{0.808} & \simresult{335}{338}{0.792} \\
$r_0=0.15$ & \simresult{1122}{1088}{0.806} & \simresult{1477}{1427}{0.836} & \simresult{1122}{1180}{0.810} & \simresult{1123}{1086}{0.776} \\
$r_0=0.30$ & \simresult{923}{923}{0.772} & \simresult{1087}{1073}{0.780} & \simresult{885}{890}{0.808} & \simresult{879}{912}{0.798} \\
$r_0=0.50$ & \simresult{935}{923}{0.792} & \simresult{987}{1038}{0.790} & \simresult{849}{833}{0.828} & \simresult{836}{863}{0.814} \\
\bottomrule
\end{tabular}
}
\end{table}

\begin{table}[!htbp]
\centering
\scriptsize
\caption{Estimated-ratio binary-outcome simulations with mixed covariates.}
\label{tab:app-est-binary-mixed}
\setlength{\tabcolsep}{2pt}
\resizebox{\textwidth}{!}{%
\begin{tabular}{l*{12}{c}}
\simtableheader
Baseline & \simresult{928}{949}{0.790} & \simresult{1094}{1070}{0.814} & \simresult{889}{928}{0.820} & \simresult{883}{836}{0.778} \\
$\gamma_0=\gamma_1=0.20$ & \simresult{2865}{2760}{0.788} & \simresult{3323}{3215}{0.818} & \simresult{864}{857}{0.828} & \simresult{863}{871}{0.822} \\
$\gamma_0=\gamma_1=0.35$ & \simresult{925}{955}{0.780} & \simresult{1091}{1124}{0.832} & \simresult{888}{867}{0.786} & \simresult{882}{874}{0.798} \\
$\gamma_0=\gamma_1=0.50$ & \simresult{449}{458}{0.822} & \simresult{539}{525}{0.782} & \simresult{922}{894}{0.810} & \simresult{912}{903}{0.786} \\
$\eta_1-\eta_0=0.15$ & \simresult{944}{970}{0.780} & \simresult{1091}{1189}{0.846} & \simresult{2504}{2554}{0.808} & \simresult{2477}{2520}{0.850} \\
$\eta_1-\eta_0=0.25$ & \simresult{927}{965}{0.790} & \simresult{1093}{1050}{0.822} & \simresult{888}{878}{0.778} & \simresult{882}{904}{0.810} \\
$\eta_1-\eta_0=0.35$ & \simresult{914}{876}{0.784} & \simresult{1093}{1130}{0.810} & \simresult{446}{437}{0.812} & \simresult{445}{459}{0.808} \\
$\{q_0,q_1\}=\{0.40,0.60\}$ & \simresult{955}{924}{0.808} & \simresult{1036}{1020}{0.786} & \simresult{881}{928}{0.808} & \simresult{886}{865}{0.810} \\
$\{q_0,q_1\}=\{0.20,0.40\}$ & \simresult{961}{937}{0.778} & \simresult{1566}{1554}{0.806} & \simresult{932}{943}{0.810} & \simresult{877}{854}{0.796} \\
$\{q_0,q_1\}=\{0.60,0.80\}$ & \simresult{1473}{1503}{0.788} & \simresult{1055}{1052}{0.794} & \simresult{857}{888}{0.820} & \simresult{921}{952}{0.808} \\
$R_M^2=0$ & \simresult{754}{704}{0.820} & \simresult{907}{890}{0.812} & \simresult{852}{792}{0.798} & \simresult{845}{824}{0.790} \\
$R_M^2=0.25$ & \simresult{926}{913}{0.798} & \simresult{1092}{1055}{0.794} & \simresult{888}{862}{0.778} & \simresult{882}{856}{0.808} \\
$R_M^2=0.50$ & \simresult{1242}{1248}{0.798} & \simresult{1471}{1533}{0.802} & \simresult{927}{954}{0.810} & \simresult{922}{994}{0.782} \\
$\rho=0$ & \simresult{923}{961}{0.776} & \simresult{1122}{1091}{0.804} & \simresult{865}{862}{0.774} & \simresult{857}{910}{0.820} \\
$\rho=0.5$ & \simresult{927}{935}{0.808} & \simresult{1093}{1045}{0.840} & \simresult{889}{847}{0.774} & \simresult{883}{894}{0.792} \\
$\rho=0.8$ & \simresult{922}{918}{0.778} & \simresult{1067}{1029}{0.780} & \simresult{901}{881}{0.804} & \simresult{897}{922}{0.784} \\
$\{\gamma_0,\gamma_1\}=\{0.25,0.45\}$ & \simresult{550}{564}{0.788} & \simresult{2128}{2041}{0.818} & \simresult{404}{391}{0.766} & \simresult{512}{501}{0.784} \\
$\{\gamma_0,\gamma_1\}=\{0.15,0.55\}$ & \simresult{365}{377}{0.798} & \simresult{5892}{5835}{0.818} & \simresult{231}{238}{0.802} & \simresult{337}{353}{0.772} \\
$r_0=0.15$ & \simresult{1136}{1109}{0.794} & \simresult{1488}{1408}{0.822} & \simresult{1134}{1161}{0.812} & \simresult{1138}{1127}{0.814} \\
$r_0=0.30$ & \simresult{923}{923}{0.810} & \simresult{1089}{1089}{0.786} & \simresult{888}{915}{0.792} & \simresult{882}{816}{0.802} \\
$r_0=0.50$ & \simresult{932}{916}{0.798} & \simresult{986}{1101}{0.788} & \simresult{845}{923}{0.834} & \simresult{833}{874}{0.770} \\
\bottomrule
\end{tabular}
}
\end{table}
\section{Comparison with total-effect-based planning}
\label{app:ate-comparison}

We compare the proposed natural-effect sample sizes with the conventional
two-arm ATE calculation in the baseline setting and two representative panel
settings for each outcome type. The ATE calculation targets 80\% power for a
one-sided level-0.05 test using the marginal arm-specific outcome variances.
At the resulting sample size $N_{\mathrm{ATE}}$, each of 1000 simulated trials
is analyzed using the ordinary treatment-arm mean difference and the four
normalized RMPW estimators. The RMPW analyses use the true mediator
probability ratios and the known-ratio influence-function standard errors used
in the main oracle-ratio simulations.

\begin{table}[!htbp]
\centering
\scriptsize
\caption{Comparison of conventional total-effect and proposed mediation
designs in representative Gaussian-covariate settings. For each natural
effect, the entry is $N$ (empirical power at $N_{\mathrm{ATE}}$).}
\label{tab:app-ate-comparison}
\setlength{\tabcolsep}{3pt}
\resizebox{\textwidth}{!}{%
\begin{tabular}{llccrrrr}
\toprule
&& \multicolumn{2}{c}{Total treatment effect}
& \multicolumn{4}{c}{Natural effect: $N$ (power at $N_{\mathrm{ATE}}$)} \\
\cmidrule(lr){3-4}\cmidrule(lr){5-8}
Outcome & Scenario & $N_{\mathrm{ATE}}$ & Power
& $\delta(1)$ & $\delta(0)$ & $\zeta(1)$ & $\zeta(0)$ \\
\midrule
Continuous & Baseline
& 1187 & 0.790
& 437 (0.994) & 477 (0.988) & 5246 (0.338) & 5166 (0.339) \\
Continuous & $\gamma_0=\gamma_1=0.20$
& 2035 & 0.816
& 1757 (0.857) & 1916 (0.832) & 5066 (0.442) & 5013 (0.445) \\
Continuous & $\eta_1-\eta_0=0.04$
& 2111 & 0.803
& 436 (1.000) & 475 (1.000) & 20963 (0.199) & 20644 (0.209) \\
\addlinespace
Binary & Baseline
& 479 & 0.801
& 923 (0.538) & 1087 (0.502) & 885 (0.577) & 879 (0.576) \\
Binary & $\gamma_0=\gamma_1=0.20$
& 574 & 0.797
& 2851 (0.310) & 3302 (0.249) & 861 (0.655) & 859 (0.654) \\
Binary & $\eta_1-\eta_0=0.15$
& 1020 & 0.806
& 938 (0.818) & 1084 (0.794) & 2495 (0.500) & 2466 (0.500) \\
\bottomrule
\end{tabular}
}
\end{table}
\section{ACTG175 pilot calibration and resampling details}
\label{app:actg175-calibration}

\subsection{Pilot population and variable construction}

We use the ACTG175 data distributed with the \texttt{speff2trial}
package as a working pilot population. Treatment is coded as $A=0$ for
zidovudine monotherapy and $A=1$ for the three other treatment arms
pooled together. The mediator is the indicator
\[
M=I(\mathrm{CD4}_{20}-\mathrm{CD4}_{0}\ge c),
\]
where $c\in\{25,50,100\}$ cells per cubic millimeter and $c=50$ is the
primary threshold.

The outcome is the indicator of remaining free of the ACTG175 composite
event through 96 weeks. The composite event is the first occurrence of
a decline in CD4 count of at least 50 cells per cubic millimeter,
progression to AIDS, or death. Participants are included only if they
remain event-free through the 20-week mediator landmark and have
sufficient follow-up to determine their 96-week event status. We also
require complete information on the mediator, outcome, treatment, and
baseline covariates used in the mediator model. The resulting working
pilot population contains 1941 participants, with 468 in arm 0 and
1473 in arm 1.

The baseline covariates are age, weight, sex, race, hemophilia,
homosexual activity, history of intravenous drug use, Karnofsky score,
prior non-zidovudine antiretroviral therapy, zidovudine use during the
30 days before treatment initiation, prior zidovudine use, duration of
previous antiretroviral therapy, antiretroviral-treatment history,
symptomatic status, and baseline CD4 and CD8 counts. Age, weight,
Karnofsky score, baseline CD4 count, and baseline CD8 count are
standardized. The duration of previous antiretroviral therapy is
transformed using $\log(1+x)$ and then standardized.

Because eligibility for this working population depends partly on
post-randomization survival, mediator ascertainment, and outcome
follow-up, treatment randomization does not necessarily preserve
exchangeability within the selected landmark population. We therefore
use these data only to illustrate the calibration of design inputs and
do not interpret the resulting pilot targets as definitive causal
effects for the full randomized ACTG175 population. The reported sample
sizes refer to analyzable participants satisfying the landmark and
ascertainment criteria; translating them into baseline enrollment
targets would additionally require anticipated landmark eligibility,
mediator ascertainment, and outcome follow-up rates.

\subsection{Calibration of the design inputs}

For each mediator threshold, we fit the pooled probit model
\[
P(M=1\mid A,\mathbf X)
=
\Phi\bigl(\beta_0+\beta_A A+\mathbf X^\top\boldsymbol\beta_X\bigr).
\]
Let
\[
\widehat W_i=\widehat\beta_0+
\mathbf X_i^\top\widehat{\boldsymbol\beta}_X,
\qquad
\widehat\Delta_M=\widehat\beta_A.
\]
The mediator inputs are calibrated as
\[
\widehat q_a
=
\frac{1}{N_{\mathrm p}}
\sum_{i=1}^{N_{\mathrm p}}
\Phi(\widehat W_i+a\widehat\Delta_M),
\qquad a=0,1,
\]
and
\[
\widehat R_M^2
=
\frac{\widehat{\operatorname{Var}}(\widehat W)}
     {1+\widehat{\operatorname{Var}}(\widehat W)},
\]
where $N_{\mathrm p}=1941$ is the pilot sample size. For numerical
stability, fitted mediator probabilities used in the probability ratios
are truncated to $[10^{-4},1-10^{-4}]$.

Using these fitted probabilities, we calculate the normalized RMPW
means
\[
\widehat\theta_H(a,a')
=
\frac{
\sum_{i:A_i=a}
\widehat R_{a,a'}(M_i,\mathbf X_i)Y_i
}{
\sum_{i:A_i=a}
\widehat R_{a,a'}(M_i,\mathbf X_i)
},
\]
and obtain the pilot targets for $\delta(1)$, $\delta(0)$,
$\zeta(1)$, and $\zeta(0)$ from the corresponding contrasts of the four
nested means. Within each treatment arm, we also calculate
\[
\widehat\kappa_a
=
\widehat{\operatorname{Corr}}(Y,\widehat W\mid A=a).
\]

For the binary-outcome working model, the arm-specific parameters
$\alpha_a^\rho$, $b_a^\rho$, and $\gamma_a^\rho$ are calibrated to
match the pilot values of $\theta(a,a)$, $\kappa_a$, and $\delta(a)$.
Together with $(\widehat q_0,\widehat q_1,\widehat R_M^2)$, these
parameters determine the variance components
$d_{a,a'}$, $B_{a,a'}$, $C_{a,a'}$, and $K_a$.  Substitution of these components into the
variance formulas gives the design-stage variance $V_{\mathrm{work}}$
for each estimand. Under the observed pilot allocation
\[
\widehat\pi_1=\frac{1473}{1941},
\]
the required sample size is
\[
N_{\mathrm{work}}
=
\left\lceil
\frac{
\{z_{0.95}+z_{0.80}\}^2V_{\mathrm{work}}
}{
\tau_*^2
}
\right\rceil,
\]
where $\tau_*$ is the corresponding pilot-calibrated natural-effect
target.

\subsection{Fixed-ratio pilot resampling}

We assess finite-sample behavior at each effect-specific
$N_{\mathrm{work}}$ using 500 arm-stratified resamples. For a sample
size $N$, we draw
\[
n_1=\operatorname{round}(\widehat\pi_1N),
\qquad
n_0=N-n_1
\]
participants with replacement from the two pilot treatment arms.
The mediator probabilities fitted in the full pilot population are
carried with the sampled participants and are not refitted within a
resample. Each resample is used to recompute the four normalized RMPW
means and the corresponding natural-effect estimate.

For an estimand with pilot target $\tau_*$, rejection is recorded when
\[
\operatorname{sign}(\tau_*)
\frac{\widehat\tau^{(b)}}
{\sqrt{V_{\mathrm{work}}/N_{\mathrm{work}}}}
>
z_{0.95}.
\]
Thus, both the mediator probability ratios and the design-stage
variance are held fixed at their full-pilot calibrated values. The
reported resampling power is the proportion of the 500 resamples that
reject, and its Monte Carlo standard error is
\[
\left\{
\frac{\widehat{\operatorname{Power}}
(1-\widehat{\operatorname{Power}})}{500}
\right\}^{1/2}.
\]

\section{Continuous-mediator regression working model}
\label{app:continuous-regression}

This section collects the continuous-mediator regression approximation. The
continuous mediator case is not the primary focus of the main text, but it is
useful as a reference point because it connects the proposed causal mediation
sample size framework to familiar product-of-coefficients calculations. The
working model is not needed for defining the causal mediation functionals; it
provides a tractable design approximation under which variances can be
expressed in terms of path coefficients and residual variances.

\subsection{Regression working model and estimators}

We assume the following potential mediator model: \begin{equation} \label{M(a)}
		M(a) = f_M(\mathbf{X}) + a \Delta_M + \sigma_M \epsilon_M,
	\end{equation}
where $f_M(\mathbf{X})$ is the baseline mediator mean function, $\Delta_M$ is the treatment-induced shift in the mediator, $\sigma_M$ is the mediator noise scale, and $\epsilon_M \sim \mathcal{N}(0,1)$. Throughout the paper, for simplicity, we set $\sigma_M=1$. 
	The potential outcome working model is \begin{equation}\label{Y(a)}
		Y(a,m) =  f_a(\mathbf{X})+ \gamma_a m+ \epsilon_a,
	\end{equation}
where $f_a(\mathbf{X})$ is the treatment-specific prognostic function, $\gamma_a$ is the mediator--outcome coefficient under treatment level $a$, and $\epsilon_a$ is the outcome error satisfying $E(\epsilon_a \mid \mathbf{X}) = 0$, $E(\epsilon_a^2\mid \mathbf{X}) = \sigma^2_{a} <\infty.$ We further assume that $\epsilon_M \indep \epsilon_a \mid \mathbf{X}$, $\epsilon_M \indep \mathbf{X}$. 

Under this working model, the natural indirect effect under treatment level $a$ can be identified as \[
	\delta(a) = \theta(a,1) -\theta(a,0) = \gamma_a \Delta_M.
\]
We consider the regression-based product estimator
\[
\hat{\delta}_R(a)=\hat{\gamma}_a\hat{\Delta}_M,
\]
where $\hat{\Delta}_M$ and $\hat{\gamma}_a$ are obtained from the planned
regression-based mediation analysis.
 For design-stage sample size calculation, we consider a pre-specified
regression-based mediation analysis. Specifically, the treatment-to-mediator
path coefficient $\Delta_M$ and the arm-specific mediator-to-outcome path
coefficient $\gamma_a$ are estimated by residualized least squares under the
structural models $(\ref{M(a)})$ and $(\ref{Y(a)})$. For the natural direct effect, under the structural working
model,
\[
\zeta(a)=\theta(1,a)-\theta(0,a)
=
\theta(1,1)-\theta(0,0)-\gamma_{1-a}\Delta_M .
\]
This representation motivates the estimator
\[
\hat\zeta_R(a)
=
\bar Y_1-\bar Y_0-\hat\gamma_{1-a}\hat\Delta_M ,
\]
where $\bar Y_j$ is the sample mean of $Y$ among subjects with $A=j$.

The following theorem gives the design-stage variances of these regression-based estimators. 
 	\begin{theorem}[Variance of regression based method] \label{th1}
 	Under Assumption \ref{seq ign}, for indirect effect estimator 
 	\[
 		\sqrt{N}(\hat{\delta}_R(a) - \delta(a)) \stackrel{d}{\to} \mathcal{N}(0, V_\delta^R(a)),
 	\]
 	where
 	 \[
 		V_\delta^R(a) =\frac{\gamma_a
 			^2}{\pi(1-\pi)}+ \frac{\Delta_M^2\sigma_a^2}{\pi_a}.
 	\] For the direct effect estimator 
 	\[
\sqrt{N}(\hat{\zeta}_R(a) - \zeta(a)) \stackrel{d}{\to} \mathcal{N}(0, V_\zeta^R(a)),
 	\]
 	where\[
 		V_\zeta^R (a) =  \frac{S_a^2- \gamma_a^2 + (\gamma_a-\gamma_{1-a})^2}{\pi_a}+ \frac{S^2_{1-a}- \gamma_{1-a}^2+\Delta^2_M \sigma^2_{1-a} }{\pi_{1-a}},
 	\]
 	and $S_a^2 = \operatorname{Var}(Y(a,M(a)))$.
 \end{theorem}
 The first variance formula has the usual product-of-coefficients structure: uncertainty in the treatment-mediator path contributes through $\gamma_a^2$, while uncertainty in the mediator-outcome path contributes through $\Delta_M^2$. The second formula is not obtained by simply adding the variance of the total-effect estimator and the variance of the product term. It accounts for their covariance, because both components are estimated from the same randomized trial.
 
 Based on these variances, the required sample size for a one-sided Wald-type test of $H_0:\Delta=0$ against $H_1:\Delta>0$, under a design alternative $\Delta=\Delta_*>0$, is approximately
 \[
 	\frac{
 		\left(z_{1-\alpha}+z_{1-\beta}\right)^2 V
 	}{
 		\Delta_*^2
 	},
 \]
 where $V$ is $V_\delta^R(a)$ for the indirect effect and $V_\zeta^R(a)$ for the direct effect. A two-sided test can be handled analogously by replacing $z_{1-\alpha}$ with $z_{1-\alpha/2}$. Thus, for the same effect size, variance, power, and type-I error level, the two-sided design requires a larger sample size because $z_{1-\alpha/2}>z_{1-\alpha}$. In practice, the one-sided formula is appropriate when the direction of the mediation effect is specified in advance; otherwise, the two-sided version should be used.
 
\subsection{Restricted Baron--Kenny working model}

 	The regression-based formulas above include the classical Baron--Kenny
 working model as a special case. 
 Specifically, suppose
 \[
 M(a)=f_M(\mathbf X)+a\Delta_M+\epsilon_M,
 \]
 and
 \[
 Y(a,m)=f_Y(\mathbf X)+\zeta_{\mathrm{BK}}a+\gamma m+\epsilon_Y,
 \]
 where the mediator-to-outcome coefficient is common across treatment levels. 	Then
 \[
 \delta(0)=\delta(1)=\gamma\Delta_M.
 \]
 Moreover, under this restricted model,
 \[
 \zeta(0)=\zeta(1)=\zeta_{\mathrm{BK}}.
 \]
 Thus the regression coefficient of $A$ in the pooled outcome regression
 adjusting for $M$ and $\mathbf X$ coincides with the natural direct effect
 only because the restricted working model rules out treatment modification of
 the mediator-to-outcome path.
 
 \begin{corollary}[Variance formulas under the restricted Baron--Kenny model]
 	\label{cor:bk-model}
 	Suppose the restricted Baron--Kenny working model in the preceding remark is
 	correct, with
 	\[
 	E(\epsilon_M\mid \mathbf X)=0,\quad
 	E(\epsilon_M^2\mid \mathbf X)=1, \quad
 	E(\epsilon_Y\mid \mathbf X)=0,\quad
 	E(\epsilon_Y^2\mid \mathbf X)=\sigma_Y^2,
 	\]
 	and the usual least-squares regularity conditions hold. Let
 	$
 	\hat\delta_{\mathrm{BK}}=\hat\gamma\hat\Delta_M
 	$
 	be the product-of-coefficients estimator. Then
 	\[
 	\sqrt N
 	\left(
 	\hat\delta_{\mathrm{BK}}-\gamma\Delta_M
 	\right)
 	\stackrel{d}{\to}
 	N\left(0,V_\delta^{\mathrm{BK}}\right),
 	\]
 	where
 	\[
 	V_\delta^{\mathrm{BK}}
 	=
 	\frac{\gamma^2}{\pi(1-\pi)}
 	+
 	\Delta_M^2\sigma_Y^2 .
 	\]
 	If $\hat\zeta_{\mathrm{BK}}$ denotes the coefficient of $A$ in the pooled
 	least-squares regression of $Y$ on $A$, $M$, and $\mathbf X$, then
 	\[
 	\sqrt N
 	\left(
 	\hat\zeta_{\mathrm{BK}}-\zeta_{\mathrm{BK}}
 	\right)
 	\stackrel{d}{\to}
 	N\left(0,V_\zeta^{\mathrm{BK}}\right),
 	\]
 	where
 	\[
 	V_\zeta^{\mathrm{BK}}
 	=
 	\sigma_Y^2
 	\left\{
 	\frac{1}{\pi(1-\pi)}
 	+
 	\Delta_M^2
 	\right\}.
 	\]
 \end{corollary}

\subsection{Proof of Theorem~\ref{th1}}
 
By the Frisch--Waugh--Lovell theorem, $\hat{\Delta}_M$ and
$\hat{\gamma}_a$ can be written as residualized least-squares estimators.
For $\Delta_M$, treatment randomization gives the population residual
$A-\pi$, and for $\gamma_a$, the population residualized mediator within
arm $a$ is $M-E(M\mid X,A=a)=\epsilon_M$.
Under the correctly specified structural models and standard least-squares
regularity conditions, simple calculation gives us \[
	\sqrt{N}(\hat{\Delta}_M - \Delta_M) = \frac{1}{\sqrt{N}} \sum_{i=1}^N \frac{A_i-\pi}{\pi(1-\pi)} \epsilon_{M,i}+o_p(1) \stackrel{d}{\to} \mathcal{N}(0, V_{\Delta_M}),
\]
where $V_{\Delta_M} = \frac{1}{\pi^2(1-\pi)^2}E\left[ (A-\pi)^2 \epsilon_M^2\right]= \frac{1}{\pi(1-\pi)}$. Also \[
	\sqrt{N}(\hat{\gamma}_a- \gamma_a) = \frac{1}{\sqrt{N}} \sum_{i=1}^N \frac{I(A_i=a)}{\pi_a} \epsilon_{M,i} \epsilon_{a,i} +o_p(1) \stackrel{d}{\to} \mathcal{N}(0, V_{\gamma_a}),
\]
where $V_{\gamma_a} = \frac{1}{\pi_a^2}E\left[ I(A=a) \epsilon_M^2\epsilon_a^2\right] = \frac{\sigma_a^2}{\pi_a}$. The covariance between $\Delta_M$ and $\gamma_a$ can be computed as \[
	\operatorname{Cov} (\hat{\Delta}_M , \hat{\gamma}_a)=E \left[ \frac{A-\pi}{\pi(1-\pi)}  \epsilon_M \frac{I(A=a)}{\pi_a} \epsilon_M\epsilon_a \right]= 0,
\]
since $E(\epsilon_M^2 \epsilon_a \mid \mathbf{X}) = E(\epsilon_M^2\mid \mathbf{X}) E(\epsilon_a \mid \mathbf{X}) = 0$. Using delta method, we have \[
	V^R_{\delta}(a) = \gamma_a^2 V_{\Delta_M} +  \Delta_M^2 V _{\gamma_a} + 2\gamma_a\Delta_M \operatorname{Cov} (\hat{\Delta}_M,\hat{\gamma}_a)=\frac{\gamma_a
		^2}{\pi(1-\pi)}+ \frac{\Delta_M^2\sigma_a^2}{\pi_a}.
\]

The variance of $\hat\zeta_R(a)$ cannot be obtained by simply adding the
variance of $\bar Y_1-\bar Y_0$ and the variance of
$\hat\gamma_{1-a}\hat\Delta_M$, because the sample mean difference and
$\hat\Delta_M$ are estimated from the same randomized trial. To account for
this dependence, we use the following first-order expansion:
\[
\sqrt N\{\hat\zeta_R(a)-\zeta(a)\}
=
\frac{1}{\sqrt N}\sum_{i=1}^N
\varphi_{\zeta,a}(O_i)+o_p(1),
\]
where
\[
\begin{aligned}
	\varphi_{\zeta,a}(O_i)
	={}&
	\frac{I(A_i=1)}{\pi}
	\{Y_i-\theta(1,1)\}
	-
	\frac{I(A_i=0)}{1-\pi}
	\{Y_i-\theta(0,0)\}  \\
	&-
	\gamma_{1-a}
	\left\{
	\frac{I(A_i=1)}{\pi}\epsilon_{M,i}
	-
	\frac{I(A_i=0)}{1-\pi}\epsilon_{M,i}
	\right\} -
	\Delta_M
	\frac{I(A_i=1-a)}{\pi_{1-a}}
	\epsilon_{M,i}\epsilon_{1-a,i}.
\end{aligned}
\]
The first line is the influence function of the randomized-arm mean
difference, the second line is the contribution from estimating
$\Delta_M$ and 
$\gamma_{1-a}$.

Since $A$ is randomized, the two treatment arms contribute additively to the
asymptotic variance. In the arm $A=a$, the corresponding residual component
is, up to an irrelevant sign,
\[
Y(a,M(a))-\theta(a,a)-\gamma_{1-a}\epsilon_M .
\]
Under the structural working model,
$
\operatorname{Cov}\{Y(a,M(a)),\epsilon_M\}=\gamma_a.
$
Therefore,
\[
\begin{aligned}
	&\operatorname{Var}
	\left\{
	Y(a,M(a))-\theta(a,a)-\gamma_{1-a}\epsilon_M
	\right\}   =
	S_a^2+\gamma_{1-a}^2-2\gamma_a\gamma_{1-a}  =
	S_a^2-\gamma_a^2+(\gamma_a-\gamma_{1-a})^2 .
\end{aligned}
\]
In the arm $A=1-a$, the corresponding residual component is
\[
Y(1-a,M(1-a))-\theta(1-a,1-a)-\gamma_{1-a}\epsilon_M ,
\]
together with the additional term from estimating $\gamma_{1-a}$. Since
\[
\operatorname{Var}
\left\{
Y(1-a,M(1-a))-\theta(1-a,1-a)-\gamma_{1-a}\epsilon_M
\right\}
=
S_{1-a}^2-\gamma_{1-a}^2
\]
and
\[
\operatorname{Var}(\epsilon_M\epsilon_{1-a})=\sigma_{1-a}^2,
\]
while the corresponding covariance term is zero under the stated orthogonality
conditions, the contribution from the arm $A=1-a$ is
\[
S_{1-a}^2-\gamma_{1-a}^2+\Delta_M^2\sigma_{1-a}^2.
\]
Combining the two arms gives
\[
V_\zeta^R(a)
=
\frac{
	S_a^2-\gamma_a^2+(\gamma_a-\gamma_{1-a})^2
}{\pi_a}
+
\frac{
	S_{1-a}^2-\gamma_{1-a}^2+\Delta_M^2\sigma_{1-a}^2
}{\pi_{1-a}}.
\]

\subsection{Proof of Corollary~\ref{cor:bk-model}}
	\begin{proof}
	Under the Baron--Kenny restriction, the mediator-to-outcome coefficient does not
	depend on the treatment level. Hence,
	\[
	\theta(a,1)-\theta(a,0)
	=
	\gamma\{E[M(1)]-E[M(0)]\}
	=
	\gamma\Delta_M.
	\]
	This gives the product representation of the indirect effect.
	
	By the Frisch--Waugh--Lovell theorem, the residualized treatment in the mediator
	regression is $A-\pi$, and the residualized mediator in the outcome regression
	after adjusting for $A$ and $\mathbf X$ is $\epsilon_M$. Since
	$E(\epsilon_M^2)=1$, the least-squares influence functions are
	\[
	\frac{A-\pi}{\pi(1-\pi)}\epsilon_M
	\]
	for $\hat\Delta_M$, and
	\[
	\epsilon_M\epsilon_Y
	\]
	for $\hat\gamma$. Therefore,
	\[
	V_{\Delta_M}
	=
	\frac{1}{\pi(1-\pi)},
	\qquad
	V_\gamma
	=
	\sigma_Y^2.
	\]
	Their covariance is zero because
	\[
	E\left[
	\frac{A-\pi}{\pi(1-\pi)}
	\epsilon_M^2\epsilon_Y
	\right]
	=
	0.
	\]
	Applying the delta method to
	$\hat\delta_{\mathrm{BK}}=\hat\gamma\hat\Delta_M$ gives
	\[
	V_\delta^{\mathrm{BK}}
	=
	\gamma^2V_{\Delta_M}
	+
	\Delta_M^2V_\gamma
	=
	\frac{\gamma^2}{\pi(1-\pi)}
	+
	\Delta_M^2\sigma_Y^2.
	\]
	
	For the direct effect, consider the residualized outcome regression
	\[
	Y=f_Y(\mathbf X)+\zeta_{\mathrm{BK}}A+\gamma M+\epsilon_Y.
	\]
	After adjusting for $\mathbf X$, write $A_c=A-\pi$ and
	$M_c=\Delta_M A_c+\epsilon_M$. The design matrix for $(A_c,M_c)$ has second
	moment
	\[
	Q=
	E
	\begin{pmatrix}
		A_c^2 & A_cM_c\\
		A_cM_c & M_c^2
	\end{pmatrix}
	=
	\begin{pmatrix}
		\pi(1-\pi) & \Delta_M\pi(1-\pi)\\
		\Delta_M\pi(1-\pi) & \Delta_M^2\pi(1-\pi)+1
	\end{pmatrix}.
	\]
	Thus the first diagonal element of $Q^{-1}$, corresponding to the coefficient
	of $A$, is
	\[
	\{Q^{-1}\}_{11}
	=
	\frac{1}{\pi(1-\pi)}
	+
	\Delta_M^2.
	\]
	Since the outcome regression residual variance is $\sigma_Y^2$, we obtain
	\[
	V_\zeta^{\mathrm{BK}}
	=
	\sigma_Y^2
	\left\{
	\frac{1}{\pi(1-\pi)}
	+
	\Delta_M^2
	\right\}.
	\]
\end{proof}

\label{supp:lastpage}
\label{lastpage}

\end{document}